\documentclass[aap,preprint]{imsart}
 
\setattribute{journal}{name}{}

\usepackage{amsthm,amsmath,amsfonts,amssymb}
\usepackage[square,comma,numbers,sort&compress]{natbib}
\usepackage{units}
\usepackage{ifpdf}
\ifpdf
  \usepackage[pdftex]{graphicx}
\else
  \usepackage{graphicx}
\fi

\startlocaldefs
\newcommand{\reals}{{\mathbb R}}
\newcommand{\ints}{{\mathbb Z}}
\newcommand{\term}{\emph}

\newcommand{\prob}{\operatorname{Pr}}
\newcommand{\expect}{{\mathbb E}}
\newcommand{\dealias}{\operatorname{dealias}}
\renewcommand{\mid}{\; ; \;}
\newcommand{\fracpart}[1]{\left\langle #1 \right\rangle}
\newcommand{\sfracpart}[1]{\langle #1 \rangle}
\newcommand{\abs}[1]{\left\vert #1 \right\vert}
\newcommand{\sabs}[1]{\vert #1 \vert}

\newcommand{\round}[1]{{\left\lceil #1 \right\rfloor}}
\usepackage{mathbf-abbrevs}
\newtheorem{proposition}{Proposition}
\newtheorem{corollary}{Corollary}
\newtheorem{theorem}{Theorem}
\newtheorem{lemma}{Lemma}
\newtheorem{definition}{Definition}
\newcommand{\coef}{\operatorname{coef}}
\newcommand{\cubr}[1]{{\left\{ #1 \right\}}}
\newcommand{\scubr}[1]{{\{ #1 \}}}
\newcommand{\sign}[1]{\operatorname{sgn}\left( #1 \right)}

\endlocaldefs

\begin{document}

\begin{frontmatter}

\title{Polynomial phase estimation by phase unwrapping}
\runtitle{Polynomial phase estimation by phase unwrapping}


\begin{aug}
\author{\fnms{Robby G.} \snm{McKilliam}\thanksref{m1}\ead[label=e1]{robby.mckilliam@unisa.edu.au}},
\author{\fnms{Barry G.} \snm{Quinn}\thanksref{m2}\ead[label=e2]{barry.quinn@mq.edu.au}},
\author{\fnms{I. Vaughan L.} \snm{Clarkson}\thanksref{m3}\ead[label=e3]{v.clarkson@uq.edu.au}},
\author{\fnms{Bill} \snm{Moran}\thanksref{m4}\ead[label=e4]{wmoran@unimelb.edu.au}},
\and
\author{\fnms{Badri N.} \snm{Vellambi}\thanksref{m1}\ead[label=e5]{badri.vellambi@unisa.edu.au}}
\runauthor{R. McKilliam et. al.}

\affiliation{University of South Australia\thanksmark{m1}, Macquarie University\thanksmark{m2}, University of Queensland\thanksmark{m3} and Melbourne University\thanksmark{m4}}

\address{R. G. McKilliam and B. N. Vellambi\\
Institute for Telecommunications Research\\
University of South Australia, Adelaide\\
SA, 5095, Australia\\
\printead{e1}\\
\phantom{E-mail:\ }\printead*{e5}}

\address{B. G. Quinn\\
Department of Statistics\\
Macquarie University, Sydney,\\
NSW, 2109, Australia\\
\printead{e3}}

\address{I. V. L. Clarkson\\
School of Information Technology \& Electrical\\
Engineering, University of Queensland,\\
Brisbane, QLD, 4072, Australia\\
\printead{e3}}

\address{B. Moran\\
Department of Electrical \& Electronic\\
Engineering, Melbourne University,\\
Melbourne, VIC, 3010, Australia\\
\printead{e4}}

\end{aug}

\begin{abstract}
Estimating the coefficients of a noisy polynomial phase signal is important in fields including radar, biology and radio communications. One approach attempts to perform polynomial regression on the phase of the signal.  This is complicated by the fact that the phase is \emph{wrapped} modulo $2\pi$ and must be \emph{unwrapped} before regression can be performed. 
In this paper we consider an estimator that performs phase unwrapping in a least squares manner.  We describe the asymptotic properties of this estimator, showing that it is strongly consistent and asymptotically normally distributed. 
\end{abstract}


\begin{keyword}
\kwd{Polynomial phase signals, phase unwrapping, asymptotic statistics}
\end{keyword}

\end{frontmatter}

\section{Introduction} \label{intro}

Polynomial phase signals arise in fields including radar, sonar, geophysics, speech analysis, biology, and radio communication~\cite{Hlawatsch_lin_quad_time_freq_spmag_1992,Ridleyspeechpolyphase1989, Suga_1975_bats_echolocation, Moss_2005echolocation}. In radar and sonar applications polynomial phase signals arise when acquiring radial velocity and acceleration (and higher order motion descriptors) of a target from a reflected signal, and also in continuous wave radar and low probability of intercept radar.  In biology, polynomial phase signals are used to describe the sounds emitted by bats and dolphins for echo location.  

A polynomial phase signal of order $m$ is a function of the form
\[
s(t) = e^{2\pi j y(t)},
\]
where $j = \sqrt{-1}$, and $t$ is a real number, often representing time, and 
\[
y(t) = \tilde{\mu}_0 +\tilde{\mu}_1 t + \tilde{\mu}_2 t^2 + \dots \tilde{\mu}_m t^m
\]
is a polynomial of order $m$.  In practice the signal is typically sampled at discrete points in `time', $t$. In this paper we only consider uniform sampling, where the gap between consecutive samples is constant. In this case we can always consider the samples to be taken at some set of consecutive integers and our sampled polynomial phase signal is $s_n = s(n) = e^{2\pi j y(n)}$, where $n$ is an integer.  Of practical importance is the estimation of the coefficients $\tilde{\mu}_0, \dots, \tilde{\mu}_m$ from a number, say $N$, of consecutive observations of the noisy sampled signal
\begin{equation}\label{eq:Y_nsamplednoisey}
Y_n = \rho s_n + X_n,
\end{equation}
where $\rho$ is a real number greater than zero representing the (usually unknown) signal amplitude and $\{X_n, n \in \ints\}$ is a sequence of complex noise variables. In order to ensure identifiability it is necessary to restrict the $m+1$ coefficients to a region of $m+1$ dimensional Euclidean space $\reals^{m+1}$ called an \emph{identifiable region}.  It was shown by some of the authors \cite{McKilliam2009IndentifiabliltyAliasingPolyphase} that an identifiable region tessellates a particular $m+1$ dimensional lattice.  We discuss this in Section \ref{sec:ident_aliasing}.

An obvious estimator of the unknown coefficients is the least squares estimator.  
When $m=0$ (phase estimation) or $m=1$ (frequency estimation) the least squares estimator is an effective approach, being both computationally efficient and statistically accurate~\cite{Hannan1973,Quinn2001,McKilliam_mean_dir_est_sq_arc_length2010}. When $m \geq 2$ the computational complexity of the least squares estimator is large, and alternative estimators have been considered for this reason. These can loosely be grouped into two classes; estimators based on polynomial phase transforms, such as the discrete polynomial phase transform~\cite{Peleg_DPT_1995} and the high order phase function~\cite{Farquharson_another_poly_est_2005,Porat_asympt_HAF_DPT_1996}; and estimators based on phase unwrapping, such as Kitchen's unwrapping estimator~\cite{Kitchen_polyphase_unwrapping_1994}, and Morelande's Bayesian unwrapping estimator~\cite{Morelande_bayes_unwrapping_2009_tsp}.

In this paper we consider the estimator that results from unwrapping the phase in a least squares manner.  We call this the \emph{least squares unwrapping} (LSU) estimator.  It was shown by some of the authors \cite[Sec. 8.1]{McKilliam2010thesis}\cite{McKilliam2009asilomar_polyest_lattice} that the LSU estimator can be computed by finding a nearest point in a lattice~\cite{Agrell2002}, and Monte-Carlo simulations were used to show the LSU estimator's favourable statistical performance. 
In this paper we derive the asymptotic properties of the LSU estimator.  Under some assumptions about the distribution of the noise $X_1, \dots, X_N$, we show the estimator to be strongly consistent and asymptotically normally distributed.  Similar results were stated without proof in~\cite{McKilliam_polyphase_est_icassp_2011}.  Here, we give a proof.  The results here are also more general than in~\cite{McKilliam_polyphase_est_icassp_2011}, allowing for a wider class of noise distributions.

An interesting property is that the estimator of the $k$th polynomial phase coefficient converges almost surely to $\tilde{\mu}_k$ at rate $o(N^{-k})$.  This is perhaps not surprising, since it is the same rate observed in polynomial regression.  However, asserting that convergence at this rate occurs in the polynomial phase setting is not trivial.  For this purpose we make use of an elementary result about the number of arithmetic progressions contained inside subsets of $\{1,2,\dots,N\}$~\cite{Erdos_on_some_sequence_of_integers1936,Szemeredi_setint_no_k_arth1975,Gowers_new_proof2001}.  
The proof of asymptotic normality is complicated by the fact that the objective function corresponding with the LSU estimator is not differentiable everywhere.  Empirical process techniques~\cite{Pollard_new_ways_clts_1986,Pollard_asymp_empi_proc_1989,van2009empirical,Dudley_unif_central_lim_th_1999} and results from the literature on hyperplane arrangements~\cite{Chazelle_discrepency_method_2000,Matousek_lect_disc_geom_2002} become useful here.  We are hopeful that the proof techniques developed here will be useful for purposes other than polynomial phase estimation, and in particular other applications involving data that is `wrapped' in some sense.  Potential candidates are the phase wrapped images observed in modern radar and medical imaging devices such as synthetic aperture radar and magnetic resonance imaging~\cite{Nico_phaseunwrappingSAR_2000,Friedlander_PD_phaseunwrapping_1996}.

The paper is organised in the following way. Section~\ref{sec:lattice-theory} describes some preliminary concepts from lattice theory, and in Section~\ref{sec:ident_aliasing} we use these results to describe an identifiable region for the set of polynomial phase coefficients.  These identifiability results are required in order to properly understand the statistical properties of polynomial phase estimators. In Section \ref{sec:least-squar-unwr} we describe the LSU estimator and state its asymptotic statistical properties.  Section~\ref{sec:strongconstproof} gives the proof of strong consistency and Section~\ref{sec:centlimitproof} gives the proof of asymptotic normality. 
Section~\ref{sec:simulations} describes the results of Monte Carlo simulations with the LSU estimator.  These simulations agree with the derived asymptotic properties. 

\section{Lattices}\label{sec:lattice-theory}

A \term{lattice},  $\Lambda$, is a discrete subset of points in $\reals^n$ such that
\[
   \Lambda = \{\xbf = \Bbf\ubf \mid \ubf \in \ints^d \}
\]
where $\Bbf \in \reals^{n \times d}$ is an $n \times d$ matrix of rank $d$, called the generator matrix.  If $n = d$ the lattice is said to be full rank.  Lattices are discrete Abelian groups under vector addition.  They are subgroups of the Euclidean group $\reals^n$.  Lattices naturally give rise to tessellations of $\reals^n$ by the specification of a set of coset representatives for the quotient $\reals^n / \Lambda$.  One choice for a set of coset representatives is a fundamental parallelepiped; the parallelepiped generated by the columns of a generator matrix.  Another choice is based on the Voronoi cell those points from $\reals^n$ nearest (with respect to the Euclidean norm here) to the lattice point at the origin.  It is always possible to construct a rectangular set of representatives, as the next proposition will show.  We will use these rectangular regions for describing the aliasing properties of polynomial phase signals in Section~\ref{sec:ident_aliasing}.  These rectangular regions will be important for the derivation of the asymptotic properties of the LSU estimator in Section~\ref{sec:least-squar-unwr}.

\begin{proposition}\label{prop:lattice-theory-constructing_a_rectangular_tesselating_region}
Let  $\Lambda$ be an $n$ dimensional lattice and $\Bbf \in \reals^{n\times n}$ be a generator matrix for $\Lambda$. Let $\Bbf = \Qbf\Rbf$ where $\Qbf$ is orthonormal and $\Rbf$ is upper triangular with elements $r_{ij}$.  Then the rectangular prism $\Qbf P$ where $P = \prod_{k=1}^{n}{[-\tfrac{r_{kk}}{2}, \frac{r_{kk}}{2})}$ is a set of coset representatives for $\reals^n / \Lambda$.
\end{proposition}
\begin{proof}
This result is well known~\cite[Chapter IX, Theorem IV]{Cassels_geom_numbers_1997}~\cite[Proposition 2.1]{McKilliam2010thesis}.  This result is for lattices with full rank.  A result in the general case can be obtained similarly, but is not required here.  
\end{proof}

\begin{figure}[t]
	\centering
		\includegraphics[width=0.8\linewidth]{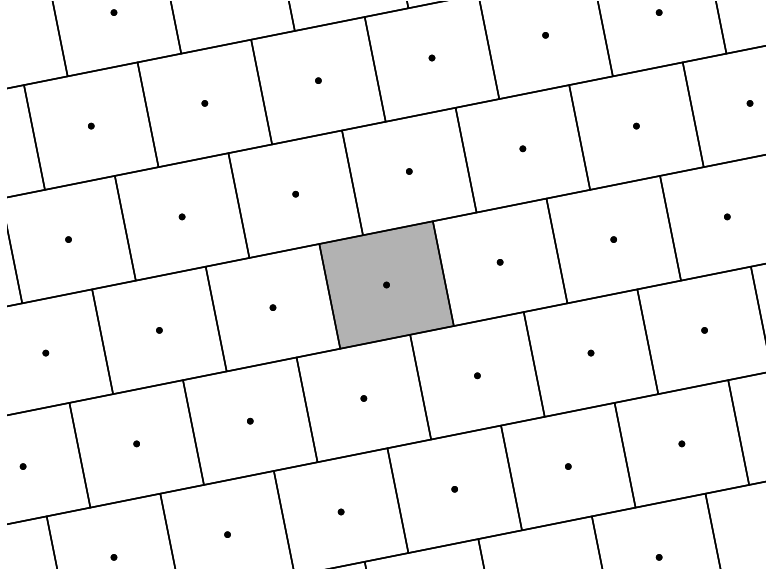}
		\caption{Rectangular tessellation constructed according to Proposition~\ref{prop:lattice-theory-constructing_a_rectangular_tesselating_region} where $\Lambda$ is a 2 dimensional lattice with generator matrix having columns $[1, 0.2]^\prime$ and $[0.2, 1]^\prime$. Any one of the boxes is a rectangular set of coset representatives for $\reals^2 / \Lambda$.  The shaded box centered at the origin is the one given by Proposition~\ref{prop:lattice-theory-constructing_a_rectangular_tesselating_region}.}
		\label{lattices:fig:tessellation2}
\end{figure}

\section{Identifiability and aliasing}\label{sec:ident_aliasing}

As discussed in the introduction, a polynomial phase signal of order $m$ is a complex valued function of the form $s(t) = e^{2\pi j y(t)}$ where $t$ is a real number and $y(t)$ is a polynomial of order $m$. We will often drop the $(t)$ and just write the polynomial as $y$ and the polynomial phase signal as $s$ whenever there is no chance of ambiguity. 
Aliasing can occur when polynomial-phase signals are sampled.  That is, two or more distinct polynomial-phase signals can take exactly the same values at the sample points.  Understanding how aliasing occurs is crucial to understanding the behaviour of polynomial phase estimators.  The aliasing properties are described in~\cite{McKilliam2009IndentifiabliltyAliasingPolyphase}, but, here we present the properties in a way that is better suited to studying the LSU estimator.  

Let $\mathcal{Z}$ be the set of polynomials of order at most $m$ that take integer values when evaluated at integers. That is, $\mathcal{Z}$ contains all polynomials $p$ such that $p(n)$ is an integer whenever $n$ is an integer.
Let $y$ and $z$ be two distinct polynomials such that $z = y + p$ for some polynomial $p$ in $\mathcal{Z}$. The two polynomial phase signals $s(t) = e^{2\pi j y(t)}$ and $r(t) = e^{2\pi j z(t)}$ are distinct because $y$ and $z$ are distinct, but if we sample $s$ and $r$ at the integers  
\begin{align*}
s(n) = e^{2\pi j y(n)} =  e^{2\pi j y(n)} e^{2\pi j p(n)} = e^{2\pi j (y(n) + p(n))} = e^{2\pi j z(n)} = r(n)
\end{align*}
because $p(n)$ is always an integer and therefore $e^{2\pi j p(n)} = 1$ for all $n \in \ints$. The polynomial phase signals $s$ and $r$ are equal at the integers, and although they are distinct, they are indistinguishable from their samples. We call such polynomial phase signals \term{aliases}\index{alias} and immediately obtain the following theorem.

\begin{theorem}\label{thm:circpolysampledthm}
Two polynomial phase signals $s(t) = e^{2\pi j y(t)}$  and $r(t) = e^{2\pi j z(t)}$  are aliases if and only if the polynomials that define their phase, $y$ and $z$, differ by a polynomial from the set $\mathcal{Z}$, that is, $y - z \in \mathcal{Z}$.
\end{theorem}



It may be helpful to observe Figures~\ref{fig:circstatplot_line},~\ref{fig:circstatplot_quad}~and~\ref{fig:circstatplot_cube}.  In these, the phase (divided by $2\pi$) of two distinct polynomial phase signals is plotted on the left, and on the right the principal component of the phase (also divided by $2\pi$) is plotted.  The circles display the samples at the integers. Note that the samples of the principal components intersect.  The corresponding polynomial phase signals are aliases.

We can derive an analogue of the theorem above in terms of the coefficients of the polynomials $y$ and $z$. This will be useful when we consider estimating the coefficients in Section~\ref{sec:least-squar-unwr}.  We first need the following family of polynomials. 

\begin{definition} \emph{(Integer valued polynomials)} \label{def:intvaledpolys}
\\The integer valued polynomial of order $k$, denoted by $p_k$, is
\[
p_k(x) = \binom{x}{k} = \frac{x(x-1)(x-2)\dots(x-k+1)}{k!},
\]
where we define $p_0(x) = 1$.
\end{definition}

\begin{lemma}\label{lem:intvalpol}
  The integer valued polynomials $p_0,\dots,p_m$ are an integer basis for $\mathcal{Z}$.  That is, every polynomial in $\mathcal{Z}$ can be uniquely written as
\begin{equation} \label{eq:lem_polynomial}
c_0 p_0 + c_1 p_1 + \dots + c_m p_m, \qquad c_0,c_1,\dots,c_m \in \ints.
\end{equation}
\end{lemma}
\begin{proof}
See~\citep[p. 2]{cahen_integer-valued_1997} or~\cite{McKilliam2009IndentifiabliltyAliasingPolyphase}. 
\end{proof}

Given a polynomial $g(x) = a_0 + a_1x + \dots + a_m x^m$, let
\[
\coef(g) = \left[ \begin{array}{ccccc} a_0 & a_1 & a_2 & \dots & a_m \end{array} \right]^\prime
\]
denote the column vector of length $m+1$ containing the coefficients of $g$.  We use superscript $^\prime$ to indicate the vector or matrix transpose.  If $y$ and $z$ differ by a polynomial from $\mathcal{Z}$ then $y = z + p$ where $p \in \mathcal{Z}$ and also $\coef(y) = \coef(z) + \coef(p)$.
  Consider the set
\[
L_{m+1} = \{ \coef(p) \mid p \in \mathcal{Z} \}
\]
containing the coefficient vectors corresponding to the polynomials in $\mathcal{Z}$.  Since the integer valued polynomials are a basis for $\mathcal{Z}$,
\begin{align*}
L_{m+1} &= \{ \coef(c_0 p_0 + c_1p_1 + \dots + c_mp_m) \mid c_i \in \ints \} \\
&= \{ c_0 \coef(p_0) + \dots + c_m\coef(p_m) \mid c_i \in \ints \}.
\end{align*}
Let
\[
\Pbf = \left[ \begin{array}{cccc} \coef(p_0)& \coef(p_1)& \dots& \coef(p_m)  \end{array} \right]
\]
be the $m+1$ by $m+1$ matrix with columns given by the coefficients of the integer valued polynomials.  Then,
\[
L_{m+1} = \{ \xbf = \Pbf\ubf \mid \ubf \in \ints^{m+1} \}
\]
and it is clear that $L_{m+1}$ is an $m+1$ dimensional lattice.  That is, the set of coefficients of the polynomials from $\mathcal{Z}$ forms a lattice with generator matrix $\Pbf$. We can restate Theorem~\ref{thm:circpolysampledthm} as:
\begin{corollary}\label{cor:circpolysampledcoef}
Two polynomial phase signals $s(t) = e^{2\pi j y(t)}$  and $r(t) = e^{2\pi j z(t)}$ are aliases if and only if $\coef(y)$ and $\coef(z)$ differ by a lattice point in $L_{m+1}$.
\end{corollary}


\begin{figure}[p]
	\centering
		\includegraphics[width=0.85\linewidth]{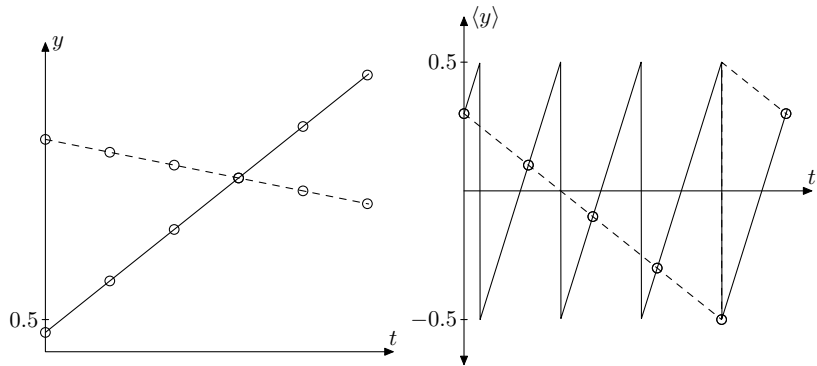}
		\caption{The first order polynomials $\tfrac{1}{10}(3 + 8t)$ (solid) and $\tfrac{1}{10}(33 - 2t)$ (dashed line).}
		\label{fig:circstatplot_line}
\end{figure}

\begin{figure}[p]
	\centering
		\includegraphics[width=0.85\linewidth]{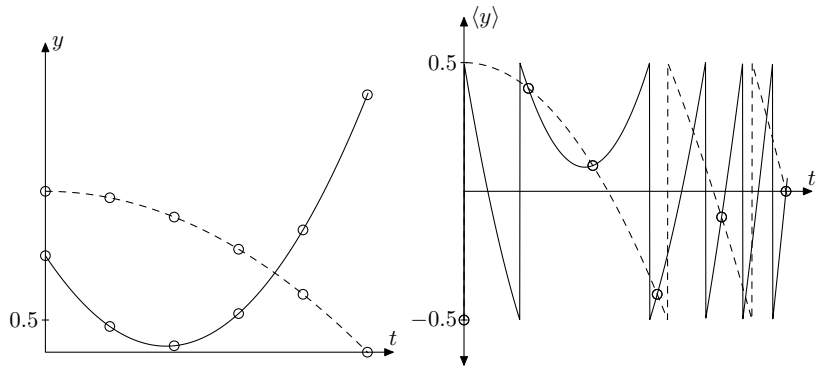}
		\caption{The quadratic polynomials $\tfrac{1}{10} (15 - 15 t + 4 t^2)$ (solid line) and $\tfrac{1}{10}(25 -  t^2)$  (dashed line).}
		\label{fig:circstatplot_quad}
\end{figure}

\begin{figure}[p]
	\centering
		\includegraphics[width=0.85\linewidth]{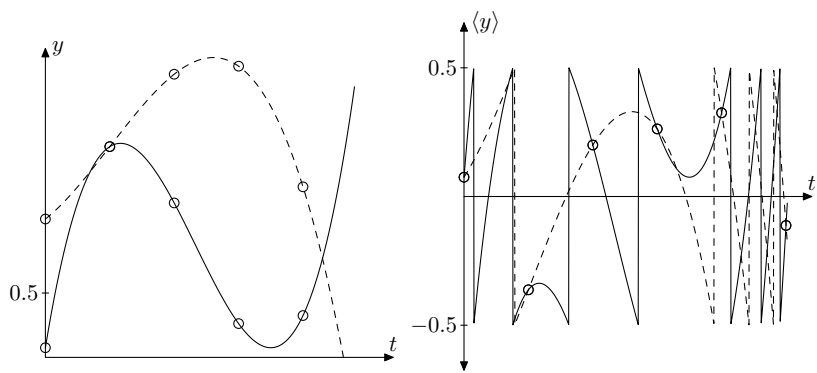}
		\caption{The cubic polynomials $\tfrac{1}{160} (174 + 85 t - 118 t^2 + 40 t^3)$ (solid line) and $\tfrac{1}{48} (84 + 19 t + 12 t^2 - 4 t^3)$  (dashed line).}
		\label{fig:circstatplot_cube}
\end{figure}

For the purpose of estimating the coefficients of a polynomial phase signal we must (in order to ensure identifiability) restrict the set of allowable coefficients so that no two polynomial phase signals are aliases of each other. In consideration of Corollary~\ref{cor:circpolysampledcoef} we require that the coefficients of $y(t)$, written in vector form $\mubf$, are contained in a set of coset representatives for the quotient $\reals^{m+1}/L_{m+1}$.  We call the chosen set of representatives the \term{identifiable region}\index{identifiable region}.

As an example consider the polynomial phase signal of order zero $e^{2\pi j \mu_0}$.  Since $e^{2\pi j \mu_0} = e^{2\pi j(\mu_0 + k)}$ for any integer $k$ we must, in order to ensure identifiability, restrict $\mu_0$ to some interval of length $1$.  A natural choice is the interval $[-\nicefrac{1}{2}, \nicefrac{1}{2})$. The lattice $L_1$ is the 1-dimensional integer lattice $\ints$ and the interval $[-\nicefrac{1}{2}, \nicefrac{1}{2})$ corresponds to the Voronoi cell of $L_1$. 
When $m=1$ it turns out that a natural choice of identifiable region is the square box $[-\nicefrac{1}{2}, \nicefrac{1}{2})^2$. This corresponds with the \term{Nyquist criterion}.  The lattice $L_2$ is equal to $\ints^2$ so the box $[-\nicefrac{1}{2}, \nicefrac{1}{2})^2$ corresponds with the Voronoi cell of $L_2$.  
When $m > 1$ the identifiable region becomes more complicated and $L_{m+1} \neq \ints^{m+1}$. 

In general there are infinitely many choices for the identifiable region. A natural choice is the Voronoi cell of $L_{m+1}$ used in~\cite{McKilliam2009IndentifiabliltyAliasingPolyphase}. Another potential choice is a fundamental parallelepiped of $L_{m+1}$. In this paper we will use the rectangular set constructed using Proposition~\ref{prop:lattice-theory-constructing_a_rectangular_tesselating_region}. Observe that $\Pbf$ is upper triangular with $k$th diagonal element equal to $\tfrac{1}{k!}$.  So this rectangular set is
\begin{equation}\label{eq:rectangular_identifiable_region}
B = \prod_{k=0}^{m}\left[ -\frac{0.5}{k!}, \frac{0.5}{k!}  \right).
\end{equation}
We will make use of this set when deriving the statistical properties of the LSU estimator in the next section. 

We define the function $\dealias(\xbf)$ to take $\xbf\in\reals^{m+1}$ to its coset representative inside $B$. That is, $\dealias(\xbf) = \zbf \in B$ where $\xbf - \zbf \in L_{m+1}$.  
When $m = 0$ or $1$ $\dealias(\xbf) = \fracpart{\xbf}$ where $\fracpart{\xbf} = \xbf - \round{\xbf}$ denotes the (centered) fractional part and $\round{\xbf}$ denotes the nearest integer to $\xbf$ with half integers rounded upwards and both $\fracpart{\cdot}$ and $\round{\cdot}$ operate on vectors elementwise.  For $m \geq 2$ the function $\dealias(\xbf)$ can be computed by a simple sequential algorithm \cite[Sec. 7.2.1]{McKilliam2010thesis}.

\section{The least squares unwrapping estimator}\label{sec:least-squar-unwr}

We now describe the least squares unwrapping (LSU) estimator of the polynomial coefficients. Recall that we desire to estimate the coefficients $\tilde{\mu}_0, \dots, \tilde{\mu}_m$ from the noisy samples $Y_1, \dots, Y_N$ given in~\eqref{eq:Y_nsamplednoisey}.  We take the complex argument of $Y_n$ and divide by $2\pi$ to obtain
\begin{equation}\label{eq:noise_circ_poly}
\Theta_n = \frac{\angle{Y_n}}{2\pi} = \fracpart{ \Phi_n + y(n) }
\end{equation}
where $\angle$ denotes the complex argument (or phase), and 
\[
\Phi_n = \frac{1}{2\pi}\angle(1 + \rho^{-1}s_n^{-1}X_n)
\] 
is a random variable representing the `phase noise' induced by $X_n$.  If the distribution of $X_n$ is circularly symmetric (the phase $\angle X_n$ is uniformly distributed on $[-\pi, \pi)$ and is independent of the magnitude $\abs{X_n}$) then the distribution of $\Phi_n$ is the same as the distribution of $\tfrac{1}{2\pi}\angle(1 + \rho^{-1}X_n)$.  If $X_1, \dots, X_N$ are circularly symmetric and identically distributed, then $\Phi_1, \dots, \Phi_n$ are also identically distributed.

Let $\mubf$ be the vector $[\mu_0, \mu_1, \dots, \mu_m]$ and put,
\begin{equation} \label{eq:sumofsquaresfunction}
SS(\mubf) = \sum_{n=1}^{N}\fracpart{  \Theta_{n} - \sum_{k = 0}^{m}{\mu_k n^k} }^{2}.
\end{equation}
The least squares unwrapping estimator is defined as those coefficients $\widehat{\mubf} = [\widehat{\mu}_0, \dots, \widehat{\mu}_m]$ that minimise $SS$ over the identifiable region $B$, i.e., the LSU estimator is,
\begin{equation}\label{eq:hatmubfLSUdefn}
\widehat{\mubf} = \arg \min_{\mubf \in B} SS(\mubf). 
\end{equation}

It is shown in~\cite[Sec~8.1]{McKilliam2010thesis}\cite{McKilliam2009asilomar_polyest_lattice} how this minimisation problem can be posed as that of computing a nearest lattice point in a particular lattice. Polynomial time algorithms that compute the nearest point are described in~\cite[Sec.~4.3]{McKilliam2010thesis}.  Although polynomial in complexity, these algorithms are not fast in practice.  The existence of practically fast nearest point algorithms for these lattices is an interesting open problem.  In this paper we focus on the asymptotic statistical properties of the LSU estimator, rather than computational aspects.

The next theorem describes the asymptotic properties of the LSU estimator.  Before we state the theorem it is necessary to understand some of the properties of the phase noise $\Phi_1,\dots,\Phi_N$, which are \emph{circular} random variables with support on $[-\nicefrac{1}{2}, \nicefrac{1}{2})$~\cite{McKilliam2010thesis,McKilliam_mean_dir_est_sq_arc_length2010,Mardia_directional_statistics,Fisher1993}.  Circular random variables are often considered modulo $2\pi$ and therefore have support $[-\pi, \pi)$ with $-\pi$ and $\pi$ being identified as equivalent.  Here we instead consider circular random variables modulo 1 with support $[-\nicefrac{1}{2}, \nicefrac{1}{2})$ and with $-\nicefrac{1}{2}$ and $\nicefrac{1}{2}$ being equivalent.  This is nonstandard but it allows us to use notation such as $\round{\cdot}$ for rounding and $\fracpart{\cdot}$ for the centered fractional part in a convenient way.   

The \emph{intrinsic mean} (or \emph{Fr\'{e}chet mean}) of $\Phi_n$ is defined as~\cite{McKilliam_mean_dir_est_sq_arc_length2010,Bhattacharya_int_ext_means_2003,Bhattacharya_int_ext_means_2005},
\begin{equation}\label{eq:intrmeandefn}
 \mu_{\text{intr}}  = \arg \min_{\mu \in [-\nicefrac{1}{2}, \nicefrac{1}{2})} \expect \fracpart{\Phi_n - \mu}^2, 
\end{equation}
and the \emph{intrinsic variance} is
\[
\sigma_{\text{intr}}^2 = \expect\fracpart{\Theta - \mu_{\text{intr}}}^2 = \min_{\mu \in [-\nicefrac{1}{2}, \nicefrac{1}{2})} \expect \fracpart{\Phi_n - \mu}^2,
\]
where $\expect$ denotes the expected value.  Depending on the distribution of $\Phi_n$ the argument that minimises~\eqref{eq:intrmeandefn} may not be unique.  The set of minima is often called the~\emph{Fr\'{e}chet mean set}~\cite{Bhattacharya_int_ext_means_2003,Bhattacharya_int_ext_means_2005}.  If the minimiser is not unique we say that $\Phi_n$ has no intrinsic mean.  
Observe the following property of circular random variables with zero intrinsic mean.
\begin{proposition}\label{prop:zerointmeanzeromean}
Let $\Phi$ be a circular random variable with intrinsic mean $\mu_{\text{intr}} = 0$ and intrinsic variance $\sigma^2$.  Then $\Phi$ has zero mean and variance $\sigma^2$, that is, $\expect \Phi = 0$ and $\expect (\Phi - \expect \Phi)^2  = \sigma^2$.
\end{proposition}
\begin{proof}
Assume the proposition is false and that $\mu = \expect \Phi \neq 0$.  But, then 
\begin{align*}
\sigma^2 = \expect\fracpart{\Phi -  \mu_{\text{intr}}}^2 = \expect\fracpart{\Phi}^2 = \expect \Phi^2 >  \expect (\Phi - \mu)^2 \geq  \expect \fracpart{\Phi - \mu}^2,
\end{align*}
violating the fact that $\mu_{\text{intr}} = 0$ is the minimiser of~\eqref{eq:intrmeandefn}.
\end{proof}
We are now equipped to state the asymptotic properties of the LSU estimator.


 
\begin{theorem} \label{thm:asymp_proof} 
Let $\widehat{\mubf}$ be defined by~\eqref{eq:hatmubfLSUdefn} and put $\widehat{\lambdabf}_N = \dealias(\tilde{\mubf} - \widehat{\mubf})$.  Denote the elements of $\widehat{\lambdabf}_N$ by $\widehat{\lambda}_{0,N}, \dots, \widehat{\lambda}_{m,N}$.  Suppose $\Phi_1, \dots, \Phi_N$ are independent and identically distributed with zero intrinsic mean, intrinsic variance $\sigma^2$, and probability density function $f$, then: 
\begin{enumerate}
\item (Strong consistency) $N^k \widehat{\lambda}_{k,N}$ converges almost surely to $0$ as $N\rightarrow\infty$ for all $k = 0, 1, \dots, m$.
\item (Asymptotic normality) If $f(\fracpart{x})$ is continuous at $x = -\nicefrac{1}{2}$ and if $f(-\nicefrac{1}{2}) < 1$ then the distribution of the vector
\[
\left[
\begin{array}
[c]{cccc}%
\sqrt{N} \widehat{\lambda}_{0,N} & N \sqrt{N}\widehat{\lambda}_{1,N}  & \dots & N^m\sqrt{N} \widehat{\lambda}_{m,N}
\end{array}
\right]^\prime
\]
converges to the normal with zero mean and covariance
\[
\frac{\sigma^2}{\left(1-f( -\nicefrac{1}{2}) \right)^{2}} \Cbf^{-1},
\]
where $\Cbf$ is the $m+1$ by $m+1$ Hilbert matrix with elements $C_{ik} = 1/(i + k + 1)$ for $i,k \in \{0, 1, \dots, m\}$.
\end{enumerate}
\end{theorem}
The proof of Theorem~\ref{thm:asymp_proof} is contained within the next two sections. Section~\ref{sec:strongconstproof} proves strong consistency and Section~\ref{sec:centlimitproof} proves asymptotic normality.  Proofs for the case when $m=0$ were given in~\cite{McKilliam_mean_dir_est_sq_arc_length2010} and for the case when $m=1$ were given in~\cite{McKilliamFrequencyEstimationByPhaseUnwrapping2009}.  The proofs here take a similar approach, but require new techniques.  
The theorem gives conditions on $\dealias(\tilde{\mubf} - \widehat{\mubf})$ rather than directly on the difference $\tilde{\mubf} - \widehat{\mubf}$.   To see why this makes sense, consider the case when $m=0$, $\tilde{\mu}_0 = -0.5$ and $\widehat{\mu}_0 = 0.49$, so that $\tilde{\mu}_0 - \widehat{\mu}_0 = -0.99$.  However, the two phases are obviously close, since the phases $\pm 0.5$ are actually the same.  In this case $\dealias(\tilde{\mu}_0 - \widehat{\mu}_0) = \fracpart{\tilde{\mu}_0 - \widehat{\mu}_0} = 0.01$ as expected.  The same reasoning holds for $m > 0$.

The requirement that $\Phi_1, \dots, \Phi_N$ be identically distributed will typically hold only when the complex random variables $X_1, \dots, X_N$ are identically distributed and circularly symmetric.  It would be possible to drop the assumption that $\Phi_1, \dots, \Phi_N$ be identically distributed, but this complicates the theorem statement and the proof.  In the interest of simplicity we only consider the case when $\Phi_1, \dots, \Phi_N$ are identically distributed here.  If $X_n$ is circularly symmetric with density function nonincreasing with magnitude $\abs{X_n}$, then the corresponding $\Phi_n$ necessarily has zero intrinsic mean~\cite[Theorem 5.2]{McKilliam2010thesis}.  Thus, our theorem covers commonly used distributions for $X_1, \dots, X_N$, such as the normal distribution.

Although we will not prove it here the assumption that $\Phi_1,\dots,\Phi_N$ have zero intrinsic mean is not only sufficient, but also necessary, for if $\Phi_1,\dots,\Phi_N$ have intrinsic mean $x \in [-\nicefrac{1}{2},\nicefrac{1}{2})$ with $x \neq 0$ then $\sfracpart{\widehat{\lambda}_{0,N} - x}\rightarrow 0$ almost surely as $N\rightarrow\infty$, and so $\widehat{\lambda}_{0,N}$ does not converge to zero.  On the other hand if $\Phi_1,\dots,\Phi_N$ do not have an intrinsic mean then $\widehat{\lambda}_{0,N}$ will not converge.

The proof of asymptotic normality places requirements on the probability density function $f$ of the phase noise $\Phi_1, \dots, \Phi_N$.  The requirement that $\Phi_1, \dots, \Phi_N$ have zero intrinsic mean implies $f(-\nicefrac{1}{2}) \leq 1$~\cite[Lemma~1]{McKilliam_mean_dir_est_sq_arc_length2010}, so the only case not handled is when $f(-\nicefrac{1}{2}) = 1$ or when $f(\fracpart{x})$ is discontinuous at $x = -\nicefrac{1}{2}$. In this exceptional case other expressions for the asymptotic variance can be found (similar to \cite[Theorem 3.1]{Hotz_circle_means_2011}), but this comes at a substantial increase in complexity and we have omitted them for this reason. 

\section{Proof of strong consistency}\label{sec:strongconstproof}
 Substituting \eqref{eq:noise_circ_poly} into $SS$ we obtain
 \begin{align*}
SS\left( \mubf \right) &= \sum_{n=1}^{N}\fracpart{ \fracpart{ \Phi_n + \sum_{k = 0}^{m}{\tilde{\mu}_k n^k} } - \sum_{k = 0}^{m}{\mu_k n^k} }^{2} \\
&= \sum_{n=1}^{N}\fracpart{  \Phi_n + \sum_{k = 0}^{m}{(\tilde{\mu}_k - \mu_k) n^k} }^{2}.
\end{align*}
Let $\lambdabf = \dealias(\tilde{\mubf} - \mubf) = \tilde{\mubf} - \mubf - \pbf$ where $\pbf$ is a lattice point from $L_{m+1}$. From the definition of $L_{m+1}$ we have $p_0 + p_1 n + \dots + p_{m} n^m$ an integer whenever $n$ is an integer, so
\begin{align*}
\fracpart{\sum_{k=0}^{m}\lambda_k n^k } = \fracpart{\sum_{k=0}^{m}(\tilde{\mu}_k - \mu_k - p_k) n^k } = \fracpart{\sum_{k=0}^{m}(\tilde{\mu}_k - \mu_k) n^k }.
\end{align*}
Let
\[
SS\left( \mubf \right) = \sum_{n=1}^{N}\fracpart{  \Phi_n + \sum_{k = 0}^{m}{\lambda_k n^k} }  ^{2} = N S_{N}\left( \lambdabf \right).
 \]
From the definition of the $\dealias(\cdot)$ function $\lambdabf \in B$ so the elements of $\lambdabf$ satisfy
 \begin{equation} \label{eq:identifiability}
 -\frac{0.5}{k!} \leq \lambda_k < \frac{0.5}{k!}.
 \end{equation} 
Now $\widehat{\lambdabf}_N = \dealias(\tilde{\mubf} - \widehat{\mubf})$ is the minimiser of $S_{N}$ in $B$. We shall show that $N^k\widehat{\lambda}_{k,N}\rightarrow0$ almost surely as $N\rightarrow\infty$ for all $k = 0,1, \dots, m$ and from this the proof of strong consistency follows.  Let
\[
V_N(\lambdabf) =  \expect S_N(\lambdabf) = \frac{1}{N}\sum_{n=1}^{N} \expect \fracpart{  \Phi_{n}+\sum_{k = 0}^{m}{\lambda_k n^k}}  ^{2}.
\]
It will follow that
 \begin{equation}\label{eq:SNVNunifmlln}
\sup_{\lambdabf \in B}\sabs{S_N(\lambdabf) - V_N(\lambdabf)} \rightarrow 0 \qquad \text{almost surely as $N\rightarrow\infty$}.  
 \end{equation}
This type of result has been called a \emph{uniform law of large numbers} and follows from standard techniques~\cite{Pollard_conv_stat_proc_1984}.  We give a full proof of~\eqref{eq:SNVNunifmlln} in Appendix~\ref{app:uniform-law-large}.  We now concentrate attention on the minimiser of $V_N$. Because $\Phi_n$ has zero intrinsic mean 
\begin{equation}\label{eq:Efracpartmined}
\expect \fracpart{ \Phi_n + z }^{2}
\end{equation}
is minimised uniquely at $z = 0$ for $z \in [-\nicefrac{1}{2}, \nicefrac{1}{2})$.  Since the intrinsic variance of $\Phi_n$ is $\sigma^2$, when $z = 0$,
\begin{equation}\label{eq:Efracpartphi}
\expect\fracpart{\Phi_1+z}^{2} = \expect\fracpart{\Phi_1}^{2} = \sigma^2,
\end{equation}
and so the minimum attained value is $\sigma^2$.

\begin{lemma}\label{lem:ES_Nminimisedzero}
For $\lambdabf \in B$ the function $V_N(\lambdabf)$ is minimised uniquely at $\zerobf$, the vector of all zeros.  At this minimum $V_N(\zerobf) = \sigma^2$.
\end{lemma}
\begin{proof}
Put $z(n) = \lambda_0 + \lambda_1 n + \dots + \lambda_m n^m$.  Then
\begin{align*}
V_N(\lambdabf) = \frac{1}{N}\expect\sum_{n=1}^{N}\fracpart{ \Phi_n + \sum_{k=0}^m{\lambda_k n^k} }^2 = \frac{1}{N}\sum_{n=1}^{N}\expect\fracpart{ \Phi_n + \fracpart{z(n)} }^2.
\end{align*}
We know that $\expect\fracpart{ \Phi_n + \fracpart{z(n)} }^2$ is minimised uniquely when $\fracpart{z(n)} = 0$ at which point it takes the value $\sigma^2$. Now $\fracpart{z(n)}$ is equal to zero for all integers $n$ if and only if $z \in \mathcal{Z}$, or equivalently if $\coef(z)$ is a lattice point in $L_{m+1}$. By definition $B$ contains precisely one lattice point from $L_{m+1}$, this being the origin $\zerobf$. Therefore $V_N$ is minimised uniquely at $\zerobf$, at which point it takes the value $\sigma^2$.
\end{proof}

\begin{lemma} \label{lem:ESNconv}
$\sabs{V_N(\widehat{\lambdabf}_N) - \sigma^2} \rightarrow 0$ almost surely as $N \rightarrow \infty$.
\end{lemma}
\begin{proof}
By definition $\widehat{\lambdabf}_N =\arg\min_{\lambdabf\in B} S_N(\lambdabf)$ so $0 \leq S_{N}(\zerobf) - S_N(\widehat{\lambdabf}_N)$.  Also, because $V_N$ is minimised at $\zerobf$, it follows that $0 \leq V_N(\widehat{\lambdabf}_N) - V_N(\zerobf)$.  Thus,
\begin{align*}
0 &\leq V_N(\widehat{\lambdabf}_N) - V_N(\zerobf) \\
 &\leq V_N(\widehat{\lambdabf}_N) - V_N(\zerobf) + S_{N}(\zerobf) - S_N(\widehat{\lambdabf}_N)   \\
&\leq \sabs{ V_N(\widehat{\lambdabf}_N) - S_N(\widehat{\lambdabf}_N) } + \sabs{ S_{N}(\zerobf) - V_N(\zerobf) }
\end{align*}
which converges almost surely to zero as $N\rightarrow\infty$ as a result of~\eqref{eq:SNVNunifmlln}.
\end{proof}


We have now shown that $V_N$ is uniquely minimised at $\zerobf$, that $V_N(\zerobf) = \sigma^2$, and that $V_N(\widehat{\lambdabf}_N)$ converges almost surely to $\sigma^2$.  These results are enough to show that $\widehat{\lambdabf}_N$ converges almost surely to zero.  However, this tells us nothing about the rate at which the components of $\widehat{\lambdabf}_N$ approach zero as required by Theorem~\ref{thm:asymp_proof}.  To prove these stronger properties we need some preliminary results about arithmetic progressions, and from the calculus of finite differences.
 
Let $W = \{1,2,\dots, N\}$ and let $K$ be a subset of $W$.  For any integer $h$, let
\begin{equation} \label{eq:S(h,G)def} 
A(h,K) = \big\{ n \mid n + ih \in K \;\forall\; i \in \{0,1,\dots,m\} \big\}
\end{equation}
be the set containing all integers $n$ such that the arithmetic progression
\[
n, \,\, n + h, \,\, n + 2h, \,\, \dots, \,\, n + mh
\]
of length $m+1$ is contained in the subset $K$.  If $K$ is a small subset of $W$ then $A(h,K)$ might be empty. However, the next two lemmas and the following corollary will show that if $K$ is sufficiently large then it always contains at least one arithmetic progression (for all sufficiently small $h$) and therefore $A(h,K)$ is not empty. We do not wish to claim any novelty here, the study of arithmetic progressions within subsets of $W$ has a considerable history~\cite{Erdos_on_some_sequence_of_integers1936,Szemeredi_setint_no_k_arth1975,Gowers_new_proof2001}.  In particular, Gower's~\cite[Theorem 1.3]{Gowers_new_proof2001} gives a result far stronger than we require here.  Denote by $K \backslash \{r\}$ the set $K$ with the element $r$ removed.

\begin{lemma} \label{lem:S(h,G/r)size}
Let $r \in K$.  For any $h$, removing $r$ from $K$ removes at most $m+1$ arithmetic progressions $n, n+h, \dots n+mh$ of length $m+1$.  That is,
\[
|A(h,K \backslash \{r\})| \geq |A(h,K)| - (m+1).
\]
\end{lemma}
\begin{proof}
The proof follows because there are at most $m+1$ integers, $n$, such that $n+ih = r$ for some $i \in \{0,1,\dots,m\}$.  That is, there are at most $m+1$ arithmetic progressions of type $n, n+h, \dots n+mh$ that contain $r$.
\end{proof}

 \begin{lemma} \label{lem:S(h,K)size}
 $|A(h,K)| \geq N - mh - (N - |K|)(m+1)$.
 \end{lemma}
 \begin{proof}
 Note that $|A(h,W)| = N - mh$.  The proof follows by starting with $A(h,W)$ and applying Lemma~\ref{lem:S(h,G/r)size} precisely $|W|-|K|=N-|K|$ times. That is, $K$ can be constructed by removing $N - |K|$ elements from $W$ and this removes at most $(N - |K|)(m+1)$ arithmetic progressions from $A(h,W)$.
 \end{proof}
 
 \begin{corollary} \label{cor:S(h,K)>0}
 Let $K \subseteq W$ such that $|K| > \frac{2m+1}{2m+2}N$. For all $h$ such that $1\leq h \leq\frac{N}{2m}$ the set $K$ contains at least one arithmetic progression $n, n+h, \dots, n+mh$ of length $m+1$. That is, $|A(h,K)| > 0$.
 \end{corollary}
 \begin{proof}
 By substituting the bounds $|K| > \frac{2m+1}{2m+2}N$ and $h \leq\frac{N}{2m}$ into the inequality from Lemma~\ref{lem:S(h,K)size} we immediately obtain $|A(h,K)| > 0$.
 \end{proof}

The next result we require comes from the calculus of finite differences. For any function $d(n)$ mapping $\reals$ to $\reals$, let 
\[
\Delta_h d(n) = d(n+h) - d(n)
\] 
denote the first difference with interval $h$, and let
\begin{equation}\label{eq:mthdiffformula}
\Delta_h^r d(n) = \Delta_h^{r-1} d(n+h) - \Delta_h^{r-1} d(n) = \sum_{k=0}^{r}\binom{r}{k}(-1)^{r-k}d(n+kh)
\end{equation}
denote the $r$th difference with interval $h$. Since $\sum_{k=0}^{r}\binom{r}{k} = 2^r$ it follows that $\Delta_h^r d(n)$ can be represented by adding and subtracting the 
\[
d(n), \,\, d(n+h), \,\, \dots, \,\, d(n+kh)
\] 
precisely $2^r$ times.

The operator $\Delta_h$ has special properties when applied to polynomials. If $d(n) = a_r n^r + \dots + a_0$ is a polynomial of order $r$ then
 \begin{equation} \label{eq:mfinitediffpoly}
 \Delta_h^r d(n) = h^r r! a_r. 
 \end{equation}
So, the $r$th difference of the polynomial is a constant depending on $h$, $r$ and the $r$th coefficient $a_r$~\cite[page 51]{Jordan_Calculus_of_finite_difference_1965}.  We can now continue the proof of strong consistency.  The next lemma is a key result.

\begin{lemma}\label{lem:moran2}
Suppose $\lambdabf_1, \lambdabf_2,\dots$ is a sequence of vectors from $B$ with $V_N(\lambdabf_N) - \sigma^2\rightarrow 0$ as $N\rightarrow\infty$. Then the elements $\lambda_{0,N}, \dots \lambda_{m,N}$ of $\lambdabf_N$ satisfy $N^k\lambda_{k, N}\rightarrow0$ as $N\rightarrow\infty$.
\end{lemma}
\begin{proof}
Define the function
\begin{equation}\label{eq:gz}
g(z) = \expect\fracpart{\Phi_1 + z}^2 - \sigma^2
\end{equation}
which is continuous in $z$. Because of~\eqref{eq:Efracpartmined} and~\eqref{eq:Efracpartphi}, $g(z) \geq 0$ with equality only at $z = 0$ for $z \in [-\nicefrac{1}{2}, \nicefrac{1}{2})$. Now
\[
V_N(\lambdabf_N) - \sigma^2 = \frac{1}{N}\sum_{n=1}^{N} g\left( \fracpart{ \sum_{k=0}^{m}{n^k \lambda_{k,N}} } \right) \rightarrow 0
\]
as $N \rightarrow \infty$. Let
\[
z_N(n) = \lambda_{0,N} + \lambda_{1,N} n + \lambda_{2,N} n^2 + \dots + \lambda_{m,N} n^m
\]
so that $V_N(\lambdabf_N) - \sigma^2 = \frac{1}{N}\sum_{n=1}^{N} g\left( \fracpart{z_N(n)} \right) \rightarrow 0$ as $N \rightarrow \infty$.
Choose constants 
\[
c = \frac{2m+1}{2m+2} \qquad \text{and} \qquad 0 < \delta < \frac{1}{2^{2m+1}}
\]
and define the set $K_{N}=\left\{  n\leq N \mid \sabs{\fracpart{z_N(n)}} < \delta \right\}$.  There exists $N_0$ such that for all $N > N_0$ the number of elements in $K_N$ is at least $cN$.  Too see this, suppose that $\sabs{K_N} < cN$, and let $\gamma$ be the minimum value of $g$ over $[-\nicefrac{1}{2},-\delta] \cup [\delta, \nicefrac{1}{2})$. Because $g(0) = 0$ is the unique minimiser of $g$, then $\gamma$ is strictly greater than $0$ and
\begin{align*}
V_N(\lambdabf_N) - \sigma^2 = \frac{1}{N}\sum_{n=1}^{N} g\left( \fracpart{z_N(n)} \right) \geq \frac{1}{N}\sum_{n\in K_N} \gamma = (1-c)\gamma,
\end{align*}
violating that $V_N(\lambdabf_N) - \sigma^2$ converges to zero as $N \rightarrow \infty$.  We will assume $N > N_0$ in what follows.

From Corollary \ref{cor:S(h,K)>0} it follows that for all $h$ satisfying $1\leq h \leq\frac{N}{2m}$ the set $A(h,K_N)$ contains at least one element, that is, there exists $n' \in A(h,K_N)$ such that all the elements from the arithmetic progression $n', n'+h, \dots, n' + mh$ are in $K_N$ and therefore 
\[
\sabs{\fracpart{z_N(n')}},\,\, \sabs{\fracpart{z_N(n'+h)}}, \,\, \dots, \,\, \sabs{\fracpart{z_N(n'+mh)}} 
\] 
are all less than $\delta$.  Because the $m$th difference is a linear combination of $2^{m}$ elements (see~\eqref{eq:mthdiffformula}) from 
\[
\fracpart{z_N(n')}, \,\, \fracpart{z_N(n'+h)}, \,\, \dots, \,\, \fracpart{z_N(n'+mh)}
\]
all with magnitude less than $\delta$ we obtain, from Lemma~\ref{lem:fracpartsumanddelta},
\begin{equation}\label{eq:Deltazfracbound}
|\fracpart{\Delta_h^m z_N(n')}| \leq |\Delta_h^m \fracpart{ z_N(n')}| < 2^{m}\delta.
\end{equation}
From \eqref{eq:mfinitediffpoly} it follows that the left hand side is equal to a constant involving $h$, $m$ and $\lambda_{m,N}$ giving the bound
\begin{equation}\label{eq:startiterativearg}
|\fracpart{ h^m m! \lambda_{m,N} }|  = |\fracpart{   \Delta_h^m z_N(n') }| < 2^m\delta
\end{equation}
for all $h$ satisfying $1\leq h \leq\frac{N}{2m}$. Setting $h = 1$ and recalling from~\eqref{eq:identifiability} that $\lambda_{m,N} \in [-\tfrac{0.5}{m!}, \tfrac{0.5}{m!})$, we have
 \[
 |\fracpart{ m! \lambda_{m,N} }| = | m! \lambda_{m,N} |< 2^m\delta.
 \]
Now, because we chose $\delta < \tfrac{1}{2^{2m}}$ it follows that 
\[
| \lambda_{m,N} |< \frac{2^m}{m!}\delta < \frac{1}{m! 2^{m+1}}.
\]
So, when $h = 2$, 
\[
|\fracpart{ 2^m m! \lambda_{m,N} }| = | 2^m m! \lambda_{m,N} |< 2^m\delta
\]
because $2^m m! \lambda_{m,N} \in [-0.5, 0.5)$. Therefore
\[
| \lambda_{m,N} |< \frac{1}{m!}\delta < \frac{1}{m! 2^{2m+1}}.
\]
Now, with $h = 4$, we similarly obtain 
\[
|\fracpart{ 4^m m! \lambda_{m,N} }| = | 4^m m! \lambda_{m,N} |< 2^m\delta
\]
and iterating this process we eventually obtain 
\[
| \lambda_{m,N} | < \frac{2^m}{2^{um} m!}\delta
\]
where $2^u$ is the largest power of 2 less than or equal to $\tfrac{N}{2m}$.  
By substituting $2^{u+1} > \frac{N}{2m}$ it follows that
 \begin{equation}\label{eq:enditerativearg}
 N^m|\lambda_{m,N}| < \frac{2^{2m+m}m^m}{m!}\delta
 \end{equation}
for all $N > N_0$.  As $\delta$ is arbitrary, $N^m \lambda_{m,N} \rightarrow 0$ as $N\rightarrow \infty$.

We have now shown that the highest order coefficient $\lambda_{m,N}$ converges as required. The remaining coefficients will be shown to converge by induction.  Assume that $N^k \lambda_{k,N} \rightarrow 0$ for all $k=r+1, r+2, \dots, m$, that is, assume that the $m-r$ highest order coefficients all converge as required. Let
\[
z_{N,r}(n) = \lambda_{0,N} + \lambda_{1,N} n + \lambda_{2,N} n^2 + \dots + \lambda_{r,N} n^r.
\]
Because the $m-r$ highest order coefficients converge we can write $z_N(n) = z_{N,r}(n) + \gamma_N(n)$ where $\sup_{n\in\{1,\dots,N\}}\abs{\gamma_N(n)} \rightarrow 0$ as $N\rightarrow\infty$. Now the bound from \eqref{eq:Deltazfracbound}, but applied using the $r$th difference, gives
 \begin{equation}\label{eq:zrbound}
\left|\fracpart{  \Delta_h^r z_N(n')}\right| = \left|\fracpart{  \Delta_h^r\gamma_N(n') + \Delta_h^r z_r(n') }\right| = |\fracpart{ \epsilon + h^r r! \lambda_{r,N} }| < 2^r\delta,
 \end{equation}
 where
\[
\epsilon = \Delta_h^r \gamma_N(n') \leq 2^r \sup_{n\in\{1,\dots,N\}}\abs{\gamma_N(n)} \rightarrow 0
\] 
as $N\rightarrow\infty$.  Choose $\delta$ and $\epsilon$ such that $2^r\delta < \tfrac{1}{4}$ and $|\epsilon| < \tfrac{1}{4}$.  Then, from \eqref{eq:zrbound} and from Lemma~\ref{lem:fracpartinternalsumlessdelta},
\[
\abs{\fracpart{ h^r r! \lambda_{r,N} }} < 2^r\delta + \abs{\epsilon}
\]
for all $h$ such that $1 \leq h \leq \tfrac{N}{2m}$.  Choosing $2^r\delta + |\epsilon| < \frac{1}{2^{2r+1}}$ and using the same iterative process as for the highest order coefficient $\lambda_{m,N}$  (see~\eqref{eq:startiterativearg}~to~\eqref{eq:enditerativearg}) we find that $N^r \lambda_{r,N} \rightarrow 0$ as $N\rightarrow\infty$.  The proof now follows by induction.
 \end{proof}


 \begin{lemma} \label{lem:fracpartsumanddelta}
 Let $a_1, a_2, \dots, a_r$ be $r$ real numbers such  that $\left|\fracpart{a_n}\right| < \delta$ for all $n = 1,2,\dots,r$.  Then $\left|\fracpart{\sum_{n=1}^r{a_n}}\right| < r\delta.$
 \end{lemma}
 \begin{proof}
 If $\delta > \tfrac{1}{2r}$ the proof is trivial as $\left|\fracpart{\sum_{n=1}^r{a_n}}\right| \leq \tfrac{1}{2}$ for all $a_n \in \reals$.  If $\delta \leq \tfrac{1}{2r}$ then $\fracpart{\sum_{n=1}^r{a_n}} = \sum_{n=1}^r{\fracpart{a_n}}$ and
 \[
 \left\vert\fracpart{\sum_{n=1}^r{a_n}}\right\vert = \left\vert\sum_{n=1}^r{\fracpart{a_n}}\right\vert \leq \sum_{n=1}^r{\left\vert\fracpart{a_n}\right\vert} < r\delta.
 \]
 \end{proof}

\begin{lemma} \label{lem:fracpartinternalsumlessdelta}
Let $\left|\fracpart{a + \epsilon}\right| < \delta$ where $|\epsilon| < \nicefrac{1}{4}$ and $0<\delta<\nicefrac{1}{4}$. Then $\left|\fracpart{a}\right| < \delta + |\epsilon|$.
\end{lemma}
\begin{proof}
By supposition $n - \delta < a + \epsilon < n + \delta$ for some $n \in \ints$.  Since $-\delta - \epsilon > -\tfrac{1}{2}$ and $\delta - \epsilon < \tfrac{1}{2}$, it follows that
\[
n - \tfrac{1}{2} < n - \delta - \epsilon < a < n + \delta - \epsilon < n + \tfrac{1}{2}.
\]
Hence $\fracpart{a} = a - n$ and so
\[
-\delta - \abs{\epsilon} \leq -\delta - \epsilon < \fracpart{a} < \delta - \epsilon \leq \delta + \abs{\epsilon}
\]
and $\abs{\fracpart{a}} \leq \delta + \abs{\epsilon}$.

\end{proof}

We are now in a position to complete the proof of strong consistency.  Let $A$ be the subset of the sample space on which  $V_N(\widehat{\lambdabf}_N) - \sigma^2 \rightarrow 0$ as $N\rightarrow\infty$.  From Lemma~\ref{lem:ESNconv}, the $\prob\{A\} =1$.  Let $A'$ be the subset of the sample space on which $N^k\widehat{\lambda}_{k,N} \rightarrow 0$ for $k=0,\dots,m$ as $N\rightarrow\infty$.  As a result of Lemma~\ref{lem:moran2}, $A \subseteq A'$, and so $\prob\{A'\} \geq \prob\{A\} = 1$.  Strong consistency follows.

\section{Proof of asymptotic normality}\label{sec:centlimitproof}

Let $\psibf$ be the vector with $k$th component $\psi_k = N^k \lambda_k$, $k=0, \dots, m$ and let 
\[
T_{N}(\psibf) = S_N(\lambdabf) = \frac{1}{N} \sum_{n=1}^{N} \fracpart{ \Phi_n + \sum_{k=0}^m (\tfrac{n}{N})^k \psi_k }^2.
\]
Let $\widehat{\psibf}_N$ be the vector with elements $\widehat{\psi}_{k,N} = N^k \widehat{\lambda}_{k,N}$ so that $\widehat{\psibf}_N$ is the minimiser of $T_N$.  Because each of $N^k \widehat{\lambda}_{k,N}$ converges almost surely to zero as $N \rightarrow \infty$, then $\widehat{\psibf}_{N}$ converges almost surely to $\zerobf$ as $N \rightarrow \infty$.  We want to find the asymptotic distribution of
\[
\sqrt{N}\widehat{\psibf}_{N} = 
\left[
\begin{array}
[c]{c}%
\sqrt{N} \widehat{\psi}_{0,N} \\ \sqrt{N}\widehat{\psi}_{1,N}  \\ \vdots \\ \sqrt{N} \widehat{\psi}_{m,N}
\end{array}
\right]
=
\left[
\begin{array}
[c]{c}%
\sqrt{N} \widehat{\lambda}_{0,N} \\ N\sqrt{N}\widehat{\lambda}_{1,N} \\ \vdots \\ N^m\sqrt{N} \widehat{\lambda}_{m,N}
\end{array}
\right].
\]
The proof is complicated by the fact that $T_N$ is not differentiable everywhere because $\fracpart{x}^2$ is not differentiable when $\fracpart{x} = \tfrac{1}{2}$.  This precludes the use of `standard approaches' to proving asymptotic normality that are based on the mean value theorem~\cite{vonMises_diff_stats_1947,Pollard_new_ways_clts_1986,Pollard_conv_stat_proc_1984,Pollard_asymp_empi_proc_1989}.  However, we show in Lemma~\ref{lem:diffathatpsi} that all the partial derivatives $\frac{\partial T_N}{\partial \psi_\ell}$ for $\ell = 0, \dots, m$ exist, and are equal to zero, at the minimiser $\widehat{\psibf}_N$.  Thus, putting
\begin{equation}\label{eq:Wn}
W_{n} = \round{\Phi_n + \sum_{k=0}^m (\tfrac{n}{N})^k \widehat{\psi}_{k,N}},
\end{equation}
so that 
\[
T_{N}(\psibf) = \frac{1}{N} \sum_{n=1}^{N} \big( \Phi_n + \sum_{k=0}^m (\tfrac{n}{N})^k \psi_k  - W_n \big)^2,
\] 
we have,
\begin{align*}
0  = \frac{\partial T_N}{\partial \psi_\ell}(\widehat{\psibf}_N) = \frac{2}{N}\sum_{n=1}^{N}(\tfrac{n}{N})^\ell\left( \Phi_n - W_{n}+ \sum_{k=0}^{m}(\tfrac{n}{N})^k \widehat{\psi}_{k,N}  \right),
\end{align*}
for each $\ell = 0, \dots, m$.  Now $D_{\ell,N} = K_{\ell,N}$, where $D_{\ell,N} = \frac{1}{\sqrt{N}} \sum_{n=1}^{N}(\tfrac{n}{N})^\ell \Phi_n$, and
\begin{equation}\label{eq:KellN}
K_{\ell,N} = \frac{1}{\sqrt{N}}\sum_{n=1}^{N}(\tfrac{n}{N})^\ell\left(W_{n}- \sum_{k=0}^{m}(\tfrac{n}{N})^k \widehat{\psi}_{k,N}  \right).
\end{equation}
Lemma~\ref{lem:Kconvfhalf} shows that,
\begin{equation}\label{eq:KellNconv}
K_{\ell,N} =  (h - 1) \sqrt{N} \sum_{k=0}^{m}  \widehat{\psi}_{k,N} \big( C_{\ell k} + o_P(1) \big) + o_P(1),
\end{equation}
for all $\ell = 0, \dots m$, where $C_{\ell k} =  \tfrac{1}{\ell + k + 1}$, and $h = f(-\nicefrac{1}{2})$, and $o_P(1)$ denotes a random variable converging in probability to zero as $N\rightarrow\infty$.

It is now convenient to write in vector form.  Let 
\begin{equation}\label{eq:dnandkn}
\kbf_N = 
\left[
\begin{array}{ccc}%
K_{0,N}  & \cdots & K_{m,N}
\end{array}
\right]^\prime
=
\dbf_N = \left[
\begin{array}{ccc}%
D_{0,N}  & \cdots & D_{m,N}
\end{array}
\right]^\prime.
\end{equation}
From~\eqref{eq:KellNconv},
\[
\dbf_N = \kbf_N = \sqrt{N} (h - 1)(\Cbf + o_P(1)) \widehat{\psibf}_N + o_P(1)
\]
where $o_P(1)$ here means a vector or matrix of the appropriate dimension with every element converging in probability to zero as $N\rightarrow\infty$.  Thus $\sqrt{N}\widehat{\psibf}_N$ has the same asymptotic distribution as $(h - 1)^{-1}\Cbf^{-1} \dbf_N$.  Lemma~\ref{eq:convdn} shows that $\dbf_N$ is asymptotically normally distributed with zero mean and covariance matrix $\sigma^2\Cbf$.  Thus $\sqrt{N}\widehat{\psibf}_N$ is asymptotically normal with zero mean and covariance matrix
\[
\frac{\sigma^2\Cbf^{-1}\Cbf(\Cbf^{-1})^\prime}{(1 - h)^2} = 
\frac{\sigma^2\Cbf^{-1}}{(1 - h)^2}.
\] 
It remains to prove Lemmas~\ref{lem:diffathatpsi},~\ref{lem:Kconvfhalf} and~\ref{eq:convdn}.

\begin{lemma}\label{lem:diffathatpsi}
For all $\ell = 0, \dots, m$ the partial derivatives $\frac{\partial T_N}{\partial \psi_\ell}$ exist, and are equal to zero, at the minimiser $\widehat{\psibf}_N$.  That is $\frac{\partial T_N}{\partial \psi_\ell}(\widehat{\psibf}_N) = 0$ for each $\ell = 0, \dots, m$.
\end{lemma}
\begin{proof}
The function $\fracpart{x}^2$ is differentiable everywhere except if $\fracpart{x} \neq -\frac{1}{2}$, and so $T_N$ is differentiable with respect to $\psibf$ at $\widehat{\psibf}_N$ if $\sfracpart{\Phi_n + \sum_{k=0}^m (\tfrac{n}{N})^k \widehat{\psi}_{k,N}} \neq -\tfrac{1}{2}$ for all $n = 1, \dots, N$.
This is proved in Lemma~\ref{lem:boundonhatlambda}.  So the partial derivatives $\frac{\partial T_N}{\partial \psi_\ell}$ exist for all $\ell = 0, \dots, m$ at $\widehat{\psibf}_N$.  That each of the partial derivatives is equal to zero at $\widehat{\psibf}_N$ follows since $\widehat{\psibf}_N$ is a minimiser of $T_N$.
\end{proof}

\begin{lemma}\label{lem:boundonhatlambda} $\sabs{\sfracpart{\Phi_n + \sum_{k=0}^m (\nicefrac{n}{N})^k \widehat{\psi}_{k,N}}} \leq \frac{1}{2} - \frac{1}{2N}$ for all $n = 1, \dots, N$.
\end{lemma}
\begin{proof}
To simplify our notation let $B_n = \Phi_n + \sum_{k=1}^m (\nicefrac{n}{N})^k \widehat{\psi}_{k,N}$ so that we now require to prove $\abs{\fracpart{B_n + \widehat{\psi}_{0,N}}} \leq \frac{1}{2} - \frac{1}{2N}$ for all $n = 1, \dots, N$.  From~\eqref{eq:Wn}, $W_n = \round{B_n + \widehat{\psi}_{0,N}}$, and 
\begin{align*}
T_N(\widehat{\psibf}_N) = \frac{1}{N}\sum_{n=1}^N \fracpart{B_n + \widehat{\psi}_{0,N}}^2= \frac{1}{N}\sum_{n=1}^N (B_n + \widehat{\psi}_{0,N} - W_n)^2.
\end{align*}
Since $\widehat{\psi}_{0,N}$ is the minimiser of the quadratic above,
\begin{equation}\label{eq:lamsum}
\widehat{\psi}_{0,N} = -\frac{1}{N}\sum_{n=1}^N(B_n - W_n).
\end{equation}
The proof now proceeds by contradiction.  Assume that for some $k$,
\begin{equation}\label{eq:philamfraccontra}
\fracpart{B_k + \widehat{\psi}_{0,N}} > \frac{1}{2} - \frac{1}{2N}.
\end{equation}
Let $F_n = W_n$ for all $n \neq k$ and $F_k = W_k + 1$, and let
\[
\phi = -\frac{1}{N}\sum_{n=1}^N(B_n - F_n) = \widehat{\psi}_{0,N} + \frac{1}{N}.
\]
Now,
\begin{align}
(B_k + \phi - F_k)^2 &= (B_k + \phi - W_k - 1)^2 \nonumber \\
&= (B_k + \phi - W_k)^2 - 2(B_k + \phi - W_k) + 1 \nonumber \\
&= (B_k + \phi - W_k)^2 - 2(B_k + \widehat{\psi}_{0,N} - W_k) + 1 - \frac{2}{N} \nonumber \\
&= (B_k + \phi - W_k)^2 - 2\fracpart{B_k + \widehat{\psi}_{0,N}} + 1 - \frac{2}{N} \nonumber \\
&< (B_k + \phi - W_k)^2 - \frac{1}{N}, \label{eq:Bkineq}
\end{align}
where the inequality in the last line follows from~\eqref{eq:philamfraccontra}. Let $\bbf = [ \phi,  \widehat{\psi}_{1,N} , \dots , \widehat{\psi}_{m,N}]$
be the vector of length $m+1$ with components $b_0 = \phi$ and $b_\ell = \widehat{\psi}_{\ell,N}$ for $\ell = 1 , \dots m$.  Now,
\[
N T_N(\bbf) = \sum_{n=1}^N \fracpart{B_n + \phi }^2 \leq  \sum_{n=1}^N (B_n + \phi  - F_n)^2, 
\]
and using the inequality from~\eqref{eq:Bkineq},
\begin{align*}
N T_N(\bbf) &< - \frac{1}{N} + \sum_{n=1}^N (B_n + \phi  - W_n)^2 \\
&= - \frac{1}{N} + \sum_{n=1}^N (B_n + \widehat{\psi}_{0,N} + \frac{1}{N}  - W_n)^2 \\
&= \sum_{n=1}^N (B_n + \widehat{\psi}_{0,N}  - W_n)^2 +  \frac{2}{N}\sum_{n=1}^N (B_n + \widehat{\psi}_{0,N}  - W_n)\\
&= N T_N(\widehat{\psibf}_N),
\end{align*}
because $\frac{2}{N}\sum_{n=1}^N (B_n + \widehat{\psi}_{0,N}  - W_n) = 0$ as a result of~\eqref{eq:lamsum}.  But, now $T_N(\bbf) < T_N(\widehat{\psibf}_N)$ violating the fact that $\widehat{\psibf}_N$ is a minimiser of $T_N$.  So~\eqref{eq:philamfraccontra} is false by contradiction.

If $\fracpart{B_k + \widehat{\psi}_{0,N}} < -\frac{1}{2} + \frac{1}{2N}$ for some $k$, we set $F_k = W_k - 1$ and using the same procedure as before obtain $T_N(\bbf) < T_N(\widehat{\psibf}_N)$ again.  The proof follows.
\end{proof}

\begin{lemma}\label{lem:Kconvfhalf}
With $K_{\ell,N}$ defined in~\eqref{eq:KellN}, $h = f(-\nicefrac{1}{2})$, and $C_{\ell k} =  \frac{1}{\ell + k + 1}$ we have $K_{\ell,N} = (h - 1) \sqrt{N} \sum_{k=0}^{m}  \widehat{\psi}_{k,N} \big( C_{\ell k} + o_P(1) \big) + o_P(1)$ for all $\ell = 0, \dots, m$.
\end{lemma}
\begin{proof}
Care must be taken since $\widehat{\psibf}_N$ depends on the sequence $\{ \Phi_n \}$.  For $n = 1, \dots, N$ and positive $N$, let
\begin{equation}
p_{nN}(\psibf) = \sum_{k=0}^{m}\left( \tfrac{n}{N}\right)^k \psi_k,
\end{equation}
and put $q_{n}(x) = \round{\Phi_n + x}$ and $Q(x) = \expect q_{n}(x) =  \expect q_{1}(x)$.  Let
\begin{equation}\label{eq:GNdef}
G_{\ell,N}(\psibf) = \frac{1}{\sqrt{N}} \sum_{n=1}^{N}\left( \tfrac{n}{N}\right)^\ell  \big(q_n(p_{nN}(\psibf)) - Q(p_{nN}(\psibf)) \big),
\end{equation}
and put
\begin{equation}\label{eq:hatpnN}
\widehat{p}_{nN} = p_{nN}(\widehat{\psibf}_N) = \sum_{k=0}^m \left(\tfrac{n}{N}\right)^k \widehat{\psi}_{k,N}.
\end{equation}
Now $W_n$ from~\eqref{eq:Wn} can be written as $W_n = \round{\Phi_n + \widehat{p}_{nN}} = q_{n}(\widehat{p}_{nN})$ and $K_{\ell,N}$ from~\eqref{eq:KellN} can be written as
\begin{align*}
K_{\ell,N} &= \frac{1}{\sqrt{N}}\sum_{n=1}^{N}(\tfrac{n}{N})^\ell\big( q_{n}(\widehat{p}_{nN}) - \widehat{p}_{nN}  \big) \\
&= \frac{1}{\sqrt{N}}\sum_{n=1}^{N}(\tfrac{n}{N})^\ell\big(q_{n}(\widehat{p}_{nN}) - \widehat{p}_{nN} + Q(\widehat{p}_{nN}) - Q(\widehat{p}_{nN}) \big) \\
&= G_{\ell,N}(\widehat{\psibf}_N) + H_{\ell,N},
\end{align*}
where 
\begin{equation}\label{eq:HellNdef}
H_{\ell,N} = \frac{1}{\sqrt{N}}\sum_{n=1}^{N}(\tfrac{n}{N})^\ell \big( Q(\widehat{p}_{nN}) - \widehat{p}_{nN} \big).
\end{equation}
Lemma~\ref{lem:unifprobG} in the Appendix shows that for any $\delta >0$ and $\nu > 0$ there exists an $\epsilon > 0$ such that
\[
\prob\left\{ \sup_{\|\psibf\|_{\infty} < \epsilon} \abs{ G_{\ell,N}(\psibf) } > \delta   \right\} < \nu
\]
for all positive integers $N$ and all $\ell = 0,\dots,m$, where $\|\psibf\|_\infty = \sup_{k} \abs{\psi_k}$.  Since $\widehat{\psibf}_N$ converges almost surely to zero, it follows that 
\[
\lim_{N\rightarrow\infty}\prob\left\{ \|\widehat{\psibf}_N\|_{\infty} \geq \epsilon \right\} = 0
\]
for any $\epsilon > 0$, and therefore $\prob\{ \|\widehat{\psibf}_N\|_{\infty} \geq \epsilon \} < \nu$ for all sufficiently large $N$.  Now
\begin{align*}
  \prob\left\{\abs{ G_{\ell,N}(\widehat{\psibf}_N) } > \delta \right\} &= \prob\left\{ \abs{G_{\ell,N}(\widehat{\psibf}_N)} > \delta \;, \; \|\widehat{\psibf}_N\|_{\infty} < \epsilon \right\} \\
&\hspace{1cm} + \prob\left\{ \abs{G_{\ell,N}(\widehat{\psibf}_N)} > \delta  \;, \; \|\widehat{\psibf}_N\|_{\infty} \geq \epsilon \right\} \\
&\leq \prob\left\{  \sup_{\|\psibf\|_{\infty} < \epsilon} \abs{ G_{\ell,N}(\psibf) } > \delta \right\} + \prob\left\{ \|\widehat{\psibf}_N\|_{\infty} \geq \epsilon \right\} \\
&\leq 2\nu
\end{align*}
for all sufficiently large $N$.  Since $\nu$ and $\delta$ can be chosen arbitrarily small, it follows that $G_{\ell,N}(\widehat{\psibf}_N)$ converges in probability to zero as $N\rightarrow\infty$, and therefore $K_{\ell,N} = H_{\ell,N} + o_P(1)$.  Lemma~\ref{lem:sumpNhilb} shows that
\begin{align*}
H_{\ell,N} =  (h-1)\sqrt{N} \sum_{k=0}^{m}  \widehat{\psi}_{k,N} \big(C_{\ell k} + o_P(1)\big).
\end{align*}
\end{proof}

\begin{lemma}\label{lem:sumpNhilb}
With $H_{\ell,N}$ defined in~\eqref{eq:HellNdef}, $h = f(-\nicefrac{1}{2})$, and  $C_{\ell k} = \tfrac{1}{\ell + k + 1}$ we have $H_{\ell,N} = (h-1)\sqrt{N} \sum_{k=0}^{m}  \widehat{\psi}_{k,N} \big(C_{\ell k} + o_P(1)\big)$.
\end{lemma}
\begin{proof}
If $\abs{x} < 1$, then
\[
q_n(x)  = \round{\Phi_n + x} = \begin{cases}
1, & \Phi_n + x  \geq \nicefrac{1}{2} \\
-1, & \Phi_n + x  < -\nicefrac{1}{2} \\
0, & \text{otherwise},
\end{cases}
\]
and,
\begin{align*}
Q(x) = Eq_1(x) &=  \begin{cases}
\int_{\nicefrac{1}{2} -x }^{\nicefrac{1}{2}}f(t)\,dt, &   x \geq 0 \\
-\int_{-\nicefrac{1}{2}}^{-\nicefrac{1}{2} - x }f(t)\,dt, &  x  < 0.
\end{cases} 
\end{align*}
Because $f(\fracpart{x})$ is continuous at $-\nicefrac{1}{2}$ it follows that $Q(x) = x \big( h + \zeta(x) \big)$
where $\zeta(x)$ is a function that converges to zero as $x$ converges to zero.  
Observe that $\abs{\widehat{p}_{nN}} \leq \sum_{k=0}^m\sabs{ \widehat{\psi}_{k,N}}$ and, since each of the $\widehat{\psi}_{k,N} \rightarrow 0$ almost surely as $N\rightarrow\infty$, it follows that $\widehat{p}_{nN} \rightarrow 0$ almost surely uniformly in $n = 1, \dots, N$ as $N\rightarrow\infty$.  Thus, $\zeta(\widehat{p}_{nN}) \rightarrow 0$ almost surely (and therefore also in probability) uniformly in $n = 1, \dots, N$ as $N\rightarrow\infty$.  Now,
\begin{align*}
Q(\widehat{p}_{nN}) - \widehat{p}_{nN} = \widehat{p}_{nN}\big( h - 1 + \zeta(\widehat{p}_{nN}) \big) = \widehat{p}_{nN}\big( h - 1 + o_P(1) \big), 
\end{align*}
and, using~\eqref{eq:HellNdef},
\begin{align*}
H_{\ell,N} &= \frac{1}{\sqrt{N}}\sum_{n=1}^{N}(\tfrac{n}{N})^\ell \widehat{p}_{nN}\big( h - 1 + o_P(1) \big)  \\
&=  \frac{1}{\sqrt{N}} \sum_{n=1}^{N}(\tfrac{n}{N})^\ell \sum_{k=0}^m (\tfrac{n}{N})^k \widehat{\psi}_{k,N} \big( h - 1 + o_P(1) \big)   \\
&= \sqrt{N}\sum_{k=0}^m \widehat{\psi}_{k,N} \frac{1}{N} \sum_{n=1}^{N} \frac{n^{\ell + k}}{N^{\ell + k + 1}}\big( h - 1 + o_P(1)  \big).
\end{align*}
The Riemann sum
\[
\frac{1}{N} \sum_{n=1}^{N} \frac{n^{\ell + k}}{N^{\ell + k + 1}} = \int_{0}^{1}x^{k+\ell+1}dx + o_P(1),
\]
and since the integral above evaluates to $C_{\ell k} = \frac{1}{k+\ell+1}$, we have
\[
H_{\ell,N} = (h - 1)\sqrt{N}\sum_{k=0}^m \widehat{\psi}_{k,N}\big( C_{\ell k} + o_P(1)  \big).C_{\ell k}
\]
\end{proof}

\begin{lemma}\label{eq:convdn}
The distribution of the vector $\dbf_N$, defined in~\eqref{eq:dnandkn}, converges to the multivariate normal with zero mean and covariance matrix $\sigma^2\Cbf$.
\end{lemma}
\begin{proof}
For any constant vector $\alphabf$, let
\[
z_N = \alphabf^\prime \dbf_N = \frac{1}{\sqrt{N}} \sum_{n=1}^{N} \Phi_n \sum_{\ell=0}^{m}\alpha_\ell\left( \frac{n}{N} \right)^\ell.
\]
By Lypanov's central limit theorem $z_N$ is asymptotically normally distributed with zero mean and variance
\[
\lim_{N\rightarrow\infty} \sigma^2 \frac{1}{N} \sum_{n=1}^{N}\left( \sum_{\ell=0}^{m}\alpha_\ell\left( \frac{n}{N} \right)^\ell \right)^2 = \sigma^2 \alphabf^\prime \Cbf \alphabf.
\]
By the Cram\`{e}r-Wold theorem it follows that $\dbf_N$ is asymptotically normally distributed with zero mean and covariance $\sigma^2 \Cbf$. 
\end{proof}

\section{Simulations}\label{sec:simulations} 
 
This section describes the results of Monte-Carlo simulations with the least squares unwrapping (LSU) estimator.  The sample sizes considered are $N = 10, 50, 200$ and the unknown amplitude is $\rho = 1$.  The $X_1, \dots, X_N$ are pseudorandomly generated independent and identically distributed circularly symmetric complex Gaussian random variables with variance $\sigma_c^2$.  The coefficients $\tilde{\mubf} = [\tilde{\mu}_0, \dots, \tilde{\mu}_m]$ are distributed uniformly randomly in the identifiable region $B$.  The number of replications of each experiment is $T = 2000$ to obtain estimates $\widehat{\mubf}_1, \dots, \widehat{\mubf}_T$ and the corresponding dealiased errors $\widehat{\lambdabf}_t = \dealias(\widehat{\mubf}_T - \tilde{\mubf})$ are computed.  The sample mean square error (MSE) of the $k$th coefficient is computed according to $\tfrac{1}{T}\sum_{t=1}^T \widehat{\lambda}_{k,t}^2$ where $\widehat{\lambda}_{k,t}$ is the $k$th element of $\widehat{\lambdabf}_t$. 

Figure~\ref{plot:polyest1} shows the sample MSEs obtained for a polynomial phase signal of order $m=3$.  Results are displayed results for the zeroth and third order coefficients $\widehat{\mu}_1$ and $\widehat{\mu}_3$.  The results for $\widehat{\mu}_1$ and $\widehat{\mu}_2$ lead to similar conclusions.  When $N = 10$ and $50$ the LSU estimator can be computed exactly using a general purpose algorithm for finding nearest lattice points called the sphere~decoder~\cite{Pohst_sphere_decoder_1981,Agrell2002,Viterbo_sphere_decoder_1999}.  This is displayed by the circles in the figures.  When $N=200$ the sphere decoder is computationally intractable and we instead use an approximate nearest point algorithm called the $K$-best method~\cite{Zhan2006_K_best_sphere_decoder}.  This is displayed by the dots.  For the purpose of comparison we have also plotted the results for the $K$-best method when $N = 10$ and $50$.  The asymptotic variance predicted in Theorem~\ref{thm:asymp_proof} is displayed by the dashed line.  Provided the noise variance is small enough (so that the `threshold' is avoided) the sample MSE of the LSU estimator is close to that predicted by Theorem~\ref{thm:asymp_proof}.  The Cram\'{e}r-Rao lower bound for the variance of unbiased polynomial phase estimators in Gaussian noise is also plotted using the solid line~\cite{Peleg1991_CRB_PPS_1991}.  When the noise variance is small the asymptotic variance of the LSU estimator is close to the Cram\'{e}r-Rao lower bound.

\begin{figure}[p] 
   	\centering 
  		\includegraphics[width=\linewidth]{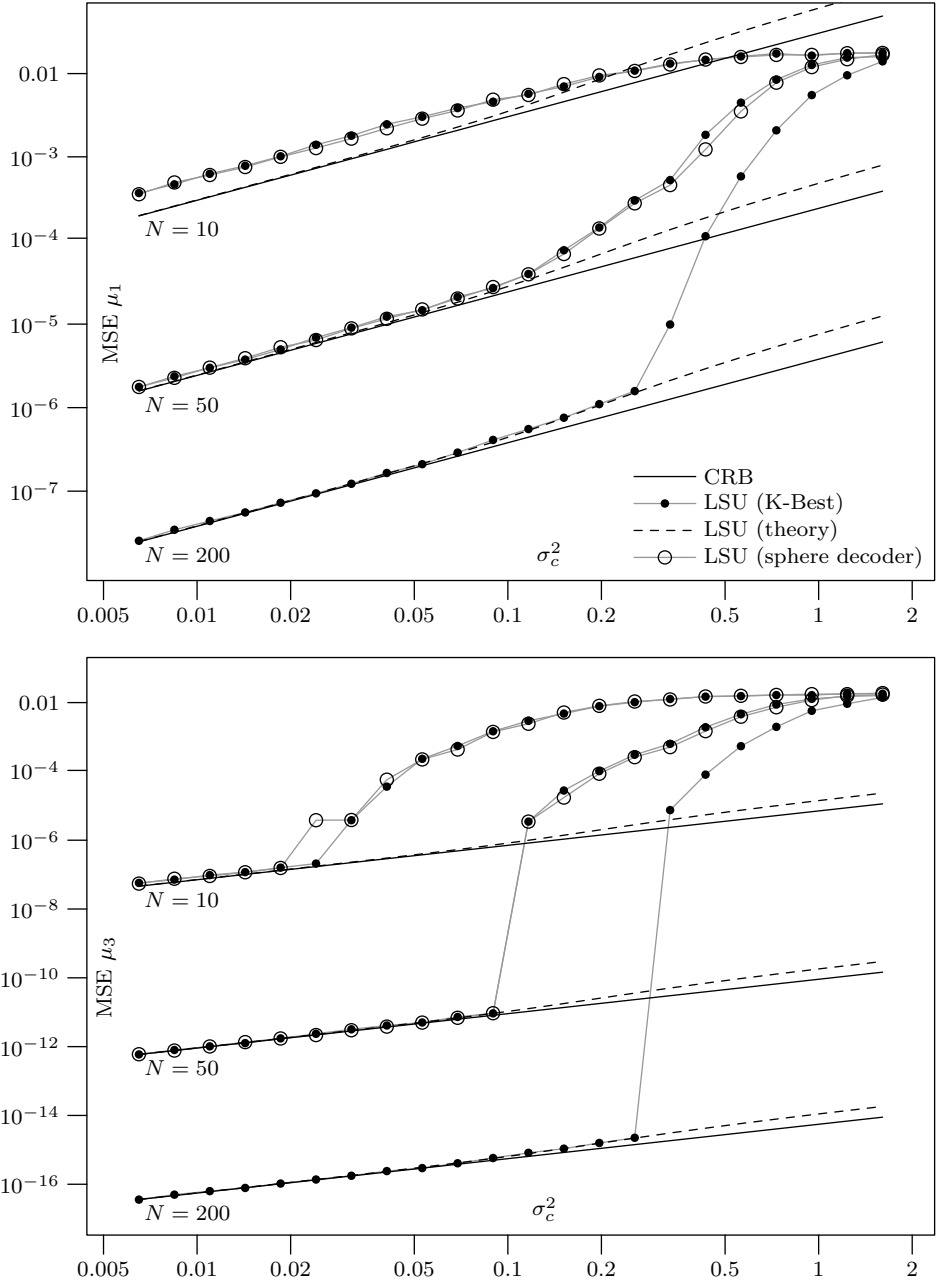}
  		\caption{Sample mean square error (MSE) of the least squares unwrapping estimator for $N=10,50$ and $200$ for a polynomial phase signal of order $m=3$. (Top) MSE of the frequency coefficient $\mu_0$.  (Bottom) MSE of the cubic coefficient $\mu_3$.} 
  		\label{plot:polyest1} 
 \end{figure}

\section{Conclusion} \label{sec:conclusion}
 
This paper has considered the estimation of the coefficients of a noisy polynomial phase signal by least squares phase unwrapping (LSU). It has been shown that the LSU estimator is strongly consistent and asymptotically normally distributed. 
Polynomial time algorithms that compute the LSU estimator are described in~\cite{McKilliam2010thesis}, but these are slow algorithms in practice.  A significant outstanding question is whether practically fast algorithms exist.  Considering the excellent statistical performance (both theoretically and practically) of the LSU estimator, even fast \emph{approximate} algorithms are likely to prove useful for the estimation of polynomial phase signals. 
  
\appendix


\section{A uniform law of large numbers} \label{app:uniform-law-large}

During the proof of strong consistency we made use of the fact that
\begin{equation}~\label{eq:SNVNunifmlln2}
\sup_{\lambdabf \in B}\sabs{S_N(\lambdabf) - V_N(\lambdabf)} \rightarrow 0
\end{equation}
almost surely as $N\rightarrow\infty$, where $V_N(\lambdabf) = \expect S_N(\lambdabf)$.  We prove this result here.  Put $D_N(\lambdabf) = S_N(\lambdabf) - V_N(\lambdabf)$.  Now $\sum_{N=1}^\infty \prob \left\{ \sup_{\lambdabf \in B}\abs{ D_N(\lambdabf) } > \epsilon \right\} < \infty$ for any $\epsilon > 0$ by Lemma~\ref{lem:vn}, and~\eqref{eq:SNVNunifmlln2} follows from the Borel-Cantelli lemma.   In what follows we use order notation in the standard way, that is, for functions $h$ and $g$, we write $h(N) = O(g(N))$ to mean that there exists a constant $K > 0$ and a finite $N_0$ such that $h(N) \leq K g(N)$ for all $N > N_0$.

\begin{lemma} \label{lem:vn} $\prob \left\{ \sup_{\lambdabf \in B}\abs{ D_N(\lambdabf) } > \epsilon \right\} = O(e^{-c \epsilon^2 N})$ for any $\epsilon > 0$ and $c < 2$.
 \end{lemma}
\begin{proof}
Consider a rectangular grid of points spaced over the identifiable region $B$.  We use $\lambdabf[\rbf]$, where $\rbf \in \ints^{m+1}$, to denote the grid point 
\[ 
 \lambdabf[\rbf] = \left[  \frac{r_0}{N^{b}} - \frac{1}{2}, \; \frac{r_1}{N^{b+1}} - \frac{1}{2}, \; \dots, \; \frac{r_m}{m!N^{b+m}} - \frac{1}{2(m!) }\right]
\]
 for some constant $b>0$.  Adjacent grid points are separated by $\tfrac{1}{N^b}$ in the zeroth coordinate, $\tfrac{1}{N^{b+1}}$ in the first coordinate and $\tfrac{1}{k!N^{b+k}}$ in the $k$th coordinate. Let
 \[
 B[\rbf]=\left\{  \xbf\in\reals^{m+1}  ; \frac{r_k}{N^{b+k}}\leq x_k + \frac{1}{2(k!)} < \frac{r_k + 1}{N^{b+k}} \right\}  .
 \]
and let $G$ be the finite set of grid points
\[
G = \left\{ \xbf \in \ints^{m+1} \mid x_k = 0,1,2\dots,N^{b+k}-1  \right\}.
\]
The total number of grid points is $|G| = N^{(m+1)(2b + m)/2}$, and the $B[\rbf]$ partition $B$, that is, $B = \cup_{\rbf \in G}B[\rbf]$.  Now 
\begin{align}
 \sup_{\lambdabf \in B}\sabs{ D_N(\lambdabf) }  &= \sup_{\rbf \in G}\sup_{ \lambdabf \in B[\rbf]}\sabs{ D_N(\lambdabf[\rbf]) + D_N(\lambdabf) - D_N(\lambdabf[\rbf])  } \nonumber \\
 &\leq \sup_{\rbf \in G}\sabs{ D_N(\lambdabf[\rbf])} + \sup_{\rbf \in G}\sup_{\lambdabf\in B[\rbf]}\sabs{ D_N(\lambdabf) - D_N(\lambdabf[\rbf])  }. \label{eq:gridanddensespace}
 \end{align}
From Lemma~\ref{lem:supVjk} it will follow that $\prob\left\{   \sup_{\rbf \in G }\abs{ D_N(\lambdabf[\rbf])  } > \frac{\epsilon}{2} \right\} = O(e^{-c \epsilon^2 N})$
for any $\epsilon > 0$ and $c < 2$.  In Lemma~\ref{lem:supBVn} we show that
\[
\sup_{\rbf \in G}\sup_{\lambdabf\in B[\rbf]}\abs{ D_N(\lambdabf) - D_N(\lambdabf[\rbf])} < 2\frac{m+1}{N^b}.
\]
Combining these results with~\eqref{eq:gridanddensespace}, we obtain
\[
\prob\left(\sup_{\lambdabf \in B }\sabs{ D_N(\lambdabf)} > \frac{\epsilon}{2} +  \frac{2(m+1)}{N^b} \right) = O(e^{-c \epsilon^2 N}),
\]
and for sufficiently large $N$, we have $\epsilon/2 + \frac{2(m+1)}{N^b} < \epsilon$ completing the proof.  It remains to prove Lemmas~\ref{lem:supVjk} and~\ref{lem:supBVn}.
\end{proof}



\begin{lemma}\label{lem:supVjk}
$\prob\cubr{   \sup_{\rbf \in G}\abs{ D_N(\lambdabf[\rbf])  } > \epsilon } = O(e^{-c \epsilon^2 N})$ for any $\epsilon > 0$ and $c < 8$.
\end{lemma}
 \begin{proof}
Fix $\lambdabf$ and write $D_N(\lambdabf) = \bar{Z} = \frac{1}{N}  \sum_{n=1}^{N}Z_{n}$, where
 \[
 Z_{n}=\fracpart{  \Phi_n+\sum_{k = 0}^{m}{\lambda_k n^k} }^{2} - \expect \fracpart{  \Phi_n + \sum_{k = 0}^{m}{\lambda_k n^k} }^{2}
 \]
 are independent with zero mean and $\abs{Z_n} \leq \tfrac{1}{4}$. It follows from Hoeffding's inequality~\cite{Hoeffding_inequality_1963} that,  $\prob\{ \abs{D_N(\lambdabf)} > \epsilon \} \leq 2e^{ - 8 \epsilon^2 N}$, and so,
 \begin{align*}
 \prob\cubr{  \sup_{\rbf\in G}\left\vert D_N(\lambdabf[\rbf]) \right\vert >\epsilon }  &\leq \sum_{\rbf \in G} \prob\scubr{  \left\vert D_N(\lambdabf[\rbf])  \right\vert >\epsilon } \nonumber \\
&= 2 |G| e^{ - 8 \epsilon^2 N} = O(e^{-c \epsilon^2 N}), \label{eq:betaprobconv}
\end{align*}
where $c$ is any real number less than $8$, since $|G| = N^{(m+1)(2b + m)/2}$ is polynomial in $N$.
\end{proof}

Before proving Lemma~\ref{lem:supBVn} we need the following result.

\begin{lemma}\label{lem:boundedsquarefracparts}
$\fracpart{x}^2 - |\delta| \leq \fracpart{x + \delta}^2 \leq \fracpart{x}^2 + |\delta|$ for all $x, \delta \in \reals$.
\end{lemma}
\begin{proof}
Since $\abs{\delta} \leq \abs{n + \delta}$ for all $\delta \in [-\nicefrac{1}{2}, \nicefrac{1}{2})$ and $n \in \ints$, the result will follow if we can show that it holds when both $x$ and $\delta$ are in $[-\nicefrac{1}{2}, \nicefrac{1}{2})$.  Also, for reasons of symmetry, we need only show that it holds when $\delta \geq 0$.  Now
\[
\fracpart{x + \delta}^2 - x^2 = \begin{cases}
2x\delta + \delta^2, & x \in [-\nicefrac{1}{2}, \nicefrac{1}{2} - \delta) \\
2x(\delta-1) + (\delta-1)^2, & x \in [\nicefrac{1}{2} - \delta, \nicefrac{1}{2})
\end{cases}
\] 
But, when $x \in [-\nicefrac{1}{2}, \nicefrac{1}{2} - \delta)$, 
\[
-1 \leq -1 + \delta \leq 2x + \delta < 1 - \delta \leq 1,
\]
and so
\[
-\delta \leq (-1 + \delta)\delta \leq (2x + \delta)\delta < (1-\delta)\delta \leq \delta.
\]
Also, when $x \in [\nicefrac{1}{2} - \delta, \nicefrac{1}{2})$ we  have $-\delta \leq 2x + \delta - 1 < \delta$, and consequently
\[
-\delta \leq -\delta(1-\delta) \leq (2x + \delta - 1)(1 - \delta) \leq \delta(1 - \delta) \leq \delta.
\]
\end{proof}

 \begin{lemma}\label{lem:supBVn} $\sup_{\rbf\in G}\sup_{ \lambdabf  \in B[\rbf] }\left\vert D_N\left(  \lambdabf \right) - D_N\left(  \lambdabf[\rbf] \right)\right\vert < 2\frac{m+1}{N^b}$ for all $N$.
 \end{lemma}
 \begin{proof}
Put $b_n = \Phi_n + \sum_{k = 0}^{m}{\lambda_k n^k}$ and $a_n = \Phi_n + \sum_{k = 0}^{m}{\lambda[\rbf]_k n^k}$, where $\lambda[\rbf]_k$ denotes the $k$th element of the grid point $\lambdabf[\rbf]$. For $\lambdabf \in B[\rbf]$ we have $b_n = a_n + \delta_n$,
where $\sabs{\delta_n} \leq \sum_{k=0}^m\frac{n^k}{k! N^{b+k}} \leq \frac{m+1}{N^{b}}$.  From Lemma~\ref{lem:boundedsquarefracparts} it follows that $-|\delta_n| \leq \fracpart{x + b_n}^2 -  \fracpart{x + a_n}^2  \leq |\delta_n|$,
and consequently $|\fracpart{x + b_n}^2 -  \fracpart{x + a_n}^2| \leq \tfrac{m+1}{N^b}$ for all $x \in \reals$. Now
\[
S_N(\lambdabf) - S_N(\lambdabf[\rbf]) =  \frac{1}{N}\sum_{n=1}^{N}\big( \fracpart{\Phi_n + b_n}^2 -  \fracpart{\Phi_n + a_n}^2 \big)
\]
and therefore $\sabs{S_N(\lambdabf) - S_N(\lambdabf[\rbf])} \leq \frac{m+1}{N^b}$ for all $\lambdabf \in B[\rbf]$.  As this bound is independent of $\Phi_1 \dots \Phi_N$, we have
\[
\sabs{ V_N(\lambdabf) - V_N(\lambdabf[\rbf])} \leq \expect \sabs{ S_N(\lambdabf) - S_N(\lambdabf[\rbf])} \leq \frac{m+1}{N^b}
\]
by Jensen's inequality.  Therefore, for all $\lambdabf \in B[\rbf]$,
 \begin{align*}
 \sabs{ D_N(\lambdabf) - D_N(\lambdabf[\rbf]) } &= \sabs{ S_N(\lambdabf) - S_N(\lambdabf[\rbf]) + V_N(\lambdabf) - V_N(\lambdabf[\rbf]) } \\
&\leq \sabs{ S_N(\lambdabf) - S_N(\lambdabf[\rbf]) \vert + \vert V_N(\lambdabf) - V_N(\lambdabf[\rbf]) } \\
&\leq 2\frac{m+1}{N^{b}},
\end{align*}
and the lemma follows because this bound is independent of $\rbf$.
 \end{proof}

\section{A tightness result}\label{app:tightness-result}

During the proof of asymptotic normality in Lemma~\ref{lem:Kconvfhalf} we made use of the following result regarding the function,
\[
G_{\ell,N}(\psibf) = \frac{1}{\sqrt{N}} \sum_{n=1}^{N}\left( \tfrac{n}{N}\right)^\ell  \big(q_n(p_{nN}(\psibf)) - Q(p_{nN}(\psibf)) \big),
\]
where the functions $q_n$, $Q$ and $p_{nN}$ are defined above~\eqref{eq:GNdef} and $\ell \in \{0, 1, \dots, m\}$.  To simplify notation we drop the subscript $\ell$ and write $G_{\ell,N}$ as $G_N$ in what follows.  The proof we will give holds for any nonnegative integers $\ell$.

\begin{lemma}\label{lem:unifprobG}
For any $\delta >0$ and $\nu > 0$ there exists $\epsilon > 0$ such that
\[
\prob\left\{ \sup_{\|\psibf\|_{\infty} < \epsilon} \abs{ G_N(\psibf) } > \delta   \right\} < \nu 
\]
for all positive integers $N$.
\end{lemma}

This result is related to what is called \emph{tightness} or \emph{asymptotic continuity} in the literature on empirical processes and weak convergence on metric spaces~\cite{Billingsley1999_convergence_of_probability_measures,Dudley_unif_central_lim_th_1999,Shorak_emp_proc_stat_2009,van2009empirical}.  The lemma is different from what is usually proved in the literature because the function $p_{nN}(\psibf) = \sum_{k=0}^m\left( \tfrac{n}{N}\right)^k\psi_k$
depends on $n$.  Nevertheless, the methods of proof from the literature can be used if we include a known result about hyperplane arrangements~\cite[Ch. 5]{Chazelle_discrepency_method_2000}\cite[Ch. 6]{Matousek_lect_disc_geom_2002}. Our proof is based on a technique called \emph{symmetrisation} and another technique called \emph{chaining} (also known as \emph{bracketing})~\cite{Pollard_asymp_empi_proc_1989,van2009empirical}.

\begin{proof}
Define the function 
\begin{align*}
f_{nN}(\psibf, \Phi_n) = \left( \tfrac{n}{N}\right)^\ell q_n(p_{nN}(\psibf)) = \left( \tfrac{n}{N}\right)^\ell \round{\Phi_n + \sum_{k=0}^m\left( \tfrac{n}{N}\right)^k\psi_k}
\end{align*}
so that $G_N$ can be written as
\[
G_N(\psibf) = \frac{1}{\sqrt{N}} \sum_{n=1}^{N} \big( f_{nN}(\psibf, \Phi_n) - \expect f_{nN}(\psibf, \Phi_n) \big).
\]
Let $\{g_n\}$ be a sequence of independent standard normal random variables, independent of the phase noise sequence $\{\Phi_n\}$.  The symmetrisation argument~\cite[Sec.~4]{Pollard_asymp_empi_proc_1989}\cite{Gine_Zinn_symmetrisation_1984,van2009empirical} can be used to show that 
\[
\expect \sup_{\|\psibf\|_{\infty} < \epsilon} \abs{ G_N(\psibf)} \leq \sqrt{2\pi} \; \expect \sup_{\|\psibf\|_{\infty} < \epsilon}  \abs{ Z_N(\psibf) },
\]
where 
\begin{equation}\label{eq:ZpsiCondGaussProc}
Z_N(\psibf) = \frac{1}{\sqrt{N}} \sum_{n=1}^{N} g_n f_{nN}(\psibf, \Phi_n),
\end{equation}
and where $\expect$ runs over both $\{g_n\}$ and $\{\Phi_n\}$.  Conditionally on $\{\Phi_n\}$, the process $\{Z_N(\psibf), \psibf \in \reals^{m+1} \}$ is a \emph{Gaussian process}, and numerous techniques exist for its analysis.  Lemma~\ref{lem:chaining} shows that for any $\kappa > 0$ there exists an $\epsilon > 0$ such that $\expect \sup_{\|\psibf\|_{\infty} < \epsilon} \abs{ Z_N(\psibf) } < \kappa.$  Thus $\expect \sup_{\|\psibf\|_{\infty} < \epsilon} \abs{ G_N(\psibf)}  <  \sqrt{2\pi} \; \kappa$, and by Markov's inequality,
\[
\prob \cubr{  \sup_{\|\psibf\|_{\infty} < \epsilon} \abs{ G_N(\psibf)} > \delta } \leq  \sqrt{2\pi} \; \frac{\kappa}{\delta},
\]
for any $\delta > 0$.  The proof follows with $\nu =  \sqrt{2\pi} \kappa/\delta$.  It remains to prove Lemma~\ref{lem:chaining}.

\end{proof}

\begin{lemma} \label{lem:chaining}
For any $\kappa > 0$ there exists $\epsilon > 0$ such that
\[
\expect \sup_{\|\psibf\|_{\infty} < \epsilon} \abs{ Z_N(\psibf) } < \kappa.
\]
\end{lemma}
\begin{proof}
Without loss of generality, assume that $\epsilon < \frac{1}{m+1}$.  Lemma~\ref{lem:chaining2} shows that
\begin{equation}\label{eq:supZCK1} 
\expect_\Phi \sup_{\|\psibf\|_{\infty} < \epsilon} \abs{ Z_N(\psibf) } \leq K_1 \sqrt{C_\epsilon(\{\Phi_n\})},
\end{equation}
where $K_1$ is a finite, positive constant, and $C_\epsilon(\{\Phi_n\})$ is the average number of times $\abs{\Phi_1}, \dots, \abs{\Phi_N}$ is greater than or equal to $\nicefrac{1}{2} - (m+1)\epsilon$.  That is,
\begin{equation}\label{eq:Cedefn}
C_\epsilon(\{\Phi_n\}) = \frac{1}{N} \sum_{n=1}^{N} I_\epsilon(\abs{\Phi_n}),
\end{equation}
where $I_\epsilon(\abs{\Phi_n})$ is 1 when $\abs{\Phi_n} \geq \nicefrac{1}{2} - (m+1)\epsilon$ and zero otherwise.  Recall that $f$ is the probability density function of $\Phi_n$, and (by assumption in Theorem~\ref{thm:asymp_proof}) that $f(\fracpart{x})$ is continuous at $x = -\nicefrac{1}{2}$.  Because of this, the expected value of $C_\epsilon(\{\Phi_n\})$ is small when $\epsilon$ is small, since
\begin{align*}
\expect C_\epsilon(\{\Phi_n\}) &= \frac{1}{N} \sum_{n=1}^{N} \expect I_\epsilon(\abs{\Phi_n}) \\
&= \prob\cubr{\abs{\Phi_1} \geq \nicefrac{1}{2} - (m+1)\epsilon} \\
&= \int_{-1/2}^{-1/2 + (m+1)\epsilon} f(\phi) d\phi + \int_{1/2 - (m+1)\epsilon}^{1/2} f(\phi) d\phi \\
&= \int_{-1/2 + (m+1)\epsilon}^{1/2 - (m+1)\epsilon} f(\fracpart{\phi}) d\phi \\
&= 2(m+1)\epsilon \big(f(-\nicefrac{1}{2}) + o(1)\big),
\end{align*}
where $o(1)$ goes to zero as $\epsilon$ goes to zero.  Since $\sqrt{\cdot}$ is a concave function on the positive real line and $C_{\epsilon}(\{\Phi_n\})$ is nonnegative, it follows from Jensen's inequality that $\expect \sqrt{C_\epsilon(\{\Phi_n\})} \leq  \sqrt{\expect  C_\epsilon(\{\Phi_n\} )} < \sqrt{K_2 \epsilon }$ for some constant $K_2$.  Applying $\expect$ to both sides of~\eqref{eq:supZCK1} gives
\[
\expect \sup_{\|\psibf\|_{\infty} < \epsilon} \abs{ Z_N(\psibf) } \leq K_1 \sqrt{\expect C_\epsilon(\{\Phi_n\})} < K_1 \sqrt{K_2 \epsilon}.
\]
Choosing $\epsilon = \kappa^2/(K_1^2 K_2)$ completes the proof.  It remains to prove Lemma~\ref{lem:chaining2}.
\end{proof}

The proofs of Lemmas~\ref{lem:chaining2}~and~\ref{lem:chaining3} are based on a technique called \emph{chaining} (or \emph{bracketing})~\cite{Dudley_unif_central_lim_th_1999,Ossiander_clt_bracketing_1984,Pollard_asymp_empi_proc_1989,Pollard_new_ways_clts_1986,van2009empirical}.  The proofs here follow those of Pollard~\cite{Pollard_asymp_empi_proc_1989}.  In the remaining lemmas we consider expectation conditional on $\{\Phi_n\}$ and treat $\{\Phi_n\}$ as a fixed realisation. We consequently use the abbreviations $C_\epsilon = C_\epsilon(\{\Phi_n\})$ and $f_{nN}(\psibf) = f_{nN}(\psibf, \Phi_n)$.  As in Lemma~\ref{lem:chaining} we assume, without loss of generality, that $\epsilon < \tfrac{1}{m+1}$.  

\begin{lemma}\label{lem:chaining2}
There exists a positive constant $K_1$ such that
\[
\expect_\Phi \sup_{\|\psibf\|_{\infty} < \epsilon} \abs{ Z_N(\psibf) } \leq K_1 \sqrt{C_\epsilon}.
\]
\end{lemma}
\begin{proof}  Let $B_\epsilon = \{\xbf \in \reals \mid \|\xbf\|_\infty < \epsilon \}$.
For each non negative integer $k$, let $T_\epsilon(k)$ be a discrete subset of $\reals^{m+1}$ with the property that for every $\psibf \in B_\epsilon$ there exists some $\psibf^* \in T_\epsilon(k)$ such that the pseudometric
\[
d(\psibf, \psibf^*) = \sum_{n = 1}^N \big( f_{nN}(\psibf) - f_{nN}(\psibf^*) \big)^2 \leq 2^{-k} C_\epsilon N.
\]
We define $T_\epsilon(0)$ to contain a single point, the origin $\zerobf$.  Defined this way $T_\epsilon(0)$ satisfies the inequality above because $d(\psibf, \zerobf) = \sum_{n = 1}^N f_{nN}(\psibf)^2  \leq C_\epsilon N$ for all $\psibf \in B_\epsilon$, as a result of Lemma~\ref{lem:epslmlemma}.

The existence of $T_\epsilon(k)$ for each positive integer $k$ will be proved in Lemma~\ref{lem:metricentropy}.  It is worth giving some intuition regarding $T_\epsilon(k)$.  If we place a `ball' of radius $2^{-k} C_\epsilon N$  with respect to the pseudometric $d(\cdot, \cdot)$ around each point in $T_\epsilon(k)$, then, by definition, the union of these balls is a superset of $B_\epsilon$.  The balls are said to \emph{cover} $B_\epsilon$ and $T_\epsilon(k)$ is said to form a \emph{covering} of $B_\epsilon$~\cite[Section 1.2]{Dudley_unif_central_lim_th_1999}.  The minimum number of such balls required to cover $B_\epsilon$ is called a \emph{covering number} of $B_\epsilon$.  In Lemma~\ref{lem:metricentropy} we show that no more than $K_3 2^{(m+1)k}$ balls of radius $2^{-k} C_\epsilon N$ are required to cover $B_\epsilon$, that is $\abs{T_\epsilon(k)} \leq K_3 2^{(m+1)k}$, where $K_3$ is a constant, independent of $N$ and $\epsilon$.  

Since $f_{nN}(\psibf) = \left( \tfrac{n}{N}\right)^\ell \round{\Phi_n + p_{nN}(\psibf)}$ is a multiple of $N^{-\ell}$ for all $\psibf \in \reals^{m+1}$, it follows that $d(\psibf, \psibf^*)$ is a multiple of $N^{-2\ell}$.  When $2^k > C_\epsilon N^{1+2\ell}$ we have $0 \leq d(\psibf, \psibf^*) < 2^{-k} C_\epsilon N < N^{-2\ell}$, and so $d(\psibf, \psibf^*) = 0$, and consequently $f_{nN}(\psibf) = f_{nN}(\psibf^*)$  for every $n = 1, \dots, N$ and $Z_N(\psibf) = Z_N(\psibf^*)$.  Thus,
\[
\sup_{\|\psibf\|_{\infty} < \epsilon } \abs{ Z_N(\psibf) } = \sup_{\psibf \in T_\epsilon(k) } \abs{ Z_N(\psibf) }
\]
for all $k$ large enough that $2^{k} > C_\epsilon N^{1+2\ell}$.  So, to analyse the supremum of $Z_N(\psibf)$ over the continuous interval $B_\epsilon$ it is enough to analyse the supremum over the discrete set $T_\epsilon(k)$ for large $k$.  
Lemma~\ref{lem:chaining3} shows that 
\[
\expect_{\Phi} \sup_{\psibf \in T_\epsilon(k) } \abs{ Z_N(\psibf) } \leq \sqrt{C_\epsilon} \sum_{i=1}^{k}\frac{\sqrt{ i A_1 + A_2}}{2^{i/2}} < \infty
\]
for every positive integer $k$, where $A_1 = 18(m+1)\log 2$ and $A_2 = 18\log K_3$ are constants and $\log(\cdot)$ is the natural logarithm.  The lemma holds with $K_1 =  \sum_{i=1}^{\infty}2^{-i/2}\sqrt{ i A_1 + A_2 }$.
\end{proof}

\begin{lemma}\label{lem:epslmlemma}
For $\epsilon < \frac{1}{m+1}$ and all $\psibf \in B_\epsilon$ and $n = 1, \dots, N$,
\[
f_{nN}(\psibf)^2 \leq \abs{f_{nN}(\psibf)} \leq I_\epsilon(\abs{\Phi_n})
\]
and consequently,
\begin{equation}\label{eq:sumf2absfCe}
\sum_{n=1}^{N}f_{nN}(\psibf)^2 \leq \sum_{n=1}^{N}\abs{f_{nN}(\psibf)} \leq N C_\epsilon.
\end{equation}
\end{lemma}
\begin{proof}
Recall that $f_{nN}(\psibf) = \left( \tfrac{n}{N}\right)^\ell \round{\Phi_n + p_{nN}(\psibf)}$.  Because $\abs{\psi_i} < \epsilon$ for all $i = 0, \dots, m$, 
\begin{equation}\label{eq:pnNsmall}
\abs{p_{nN}(\psibf)} = \abs{\sum_{i=0}^{m}\big(\tfrac{n}{N}\big)^i \psi_i} \leq (m+1)\epsilon < 1.
\end{equation}
Since $\Phi_n \in [-\nicefrac{1}{2}, \nicefrac{1}{2})$, it follows that $f_{nN}(\psibf)$ equals either $-(\tfrac{n}{N})^\ell$, $(\tfrac{n}{N})^\ell$ or $0$ and so
\begin{equation}\label{eq:fnNleqfnNs}
   f_{nN}(\psibf)^2 \leq \abs{f_{nN}(\psibf)} \leq 1.
\end{equation}
Whenever $f_{nN}(\psibf) \neq 0$ we must have 
\[
\abs{\Phi_n} \geq \nicefrac{1}{2} - \abs{p_{nN}(\psibf)} \geq \nicefrac{1}{2} - (m+1)\epsilon
\]
and $I_\epsilon(\abs{\Phi_n}) = 1$.  Thus, for all $n = 1, \dots, N$,
\[
f_{nN}(\psibf)^2 \leq \abs{f_{nN}(\psibf)} \leq I_\epsilon(\abs{\Phi_n}). 
\]
Summing the terms in this inequality over $n = 1, \dots, N$ and using~\eqref{eq:Cedefn} gives~\eqref{eq:sumf2absfCe}.
\end{proof}

\begin{lemma}\label{lem:chaining3}(Chaining)
For all positive integers $k$,
\[
\expect_{\Phi} \sup_{\psibf \in T_\epsilon(k) } \abs{ Z_N(\psibf) } \leq \sqrt{C_\epsilon} \sum_{i=1}^{k}\frac{\sqrt{ i A_1 + A_2 }}{2^{i/2}},
\]
where the constants $A_1 = 18(m+1)\log 2$ and $A_2 = 18\log K_3$.
\end{lemma}
\begin{proof}
Let $b_k$ be a function that maps each $\psibf \in T_\epsilon(k)$ to $b_k(\psibf) \in T_\epsilon(k-1)$ such that  $d(\psibf,b_k(\psibf)) \leq 2^{1-k} C_\epsilon N$.  The existence of the function $b_k$ is guaranteed by the definition of $T_\epsilon(k)$.  By the triangle inequality,
\[
\abs{ Z_N(\psibf) } \leq \abs{ Z_N(b_k(\psibf)) } + \abs{ Z_N(\psibf) - Z_N(b_k(\psibf))  },
\]
and by taking supremums on both sides,
\begin{align}
\sup_{\psibf \in T_\epsilon(k) } \sabs{ Z_N(\psibf) } &\leq \sup_{\psibf \in T_\epsilon(k) } \sabs{ Z_N(b_k(\psibf)) } + \sup_{\psibf \in T_\epsilon(k) } \sabs{ Z_N(\psibf) - Z_N(b_k(\psibf))  } \nonumber \\
&\leq \sup_{\psibf \in T_\epsilon(k-1)} \sabs{ Z_N(\psibf) } + \sup_{\psibf \in T_\epsilon(k)} \sabs{ Z_N(\psibf) - Z_N(b_k(\psibf)) }, \label{eq:supZchain}
\end{align}
the last line following since $b_k(\psibf) \in T_\epsilon(k-1)$ and so 
\[
\sup_{\psibf \in T_\epsilon(k)} \abs{ Z_N(b_k(\psibf)) } \leq \sup_{\psibf \in T_\epsilon(k-1)} \abs{ Z_N(\psibf) }.
\]
Conditional on $\Phi_1, \Phi_2, \dots$ the random variable $X(\psibf) = Z_N(\psibf) - Z_N(b_k(\psibf))$ has zero mean, and is normally distributed with variance
\begin{align*}
\sigma_X^2 &= \expect_{\Phi} \frac{1}{N} \sum_{n=1}^N g_n^2 \big( f_{nN}(\psibf) - f_{nN}(b_k(\psibf)) \big)^2 \\
&= \frac{1}{N} \sum_{n=1}^N \big( f_{nN}(\psibf) - f_{nN}(b_k(\psibf)) \big)^2 = d(\psibf, b_k(\psibf)) \leq 2^{1-k} C_\epsilon,
\end{align*}
because $\expect_{\Phi} g_n^2 = 1$.  Using Lemma~\ref{lem:maxineq}, 
\begin{align*}
\expect_{\Phi} \sup_{\psibf \in T_\epsilon(k)} \sabs{X(\psibf)} \leq 3\sqrt{2^{1-k} C_\epsilon \log\sabs{T_\epsilon(k)}} \leq \sqrt{C_\epsilon}\frac{\sqrt{  k A_1 + A_2 }}{2^{k/2}}
\end{align*}
because $\log\abs{T_\epsilon(k)} \leq k(m+1)\log 2 +  \log K_3$.  Taking expectations on both sides of~\eqref{eq:supZchain} gives
\begin{align*}
\expect_{\Phi} \sup_{\psibf \in T_\epsilon(k) } \sabs{ Z_N(\psibf) } \leq \expect_{\Phi} \sup_{\psibf \in T_\epsilon(k-1)} \sabs{ Z_N(\psibf) }  + \sqrt{C_\epsilon}\frac{\sqrt{  k A_1 + A_2 }}{2^{k/2}},
\end{align*}
which involves a recursion in $k$.  By unravelling the recursion, and using the fact $T_\epsilon(0)$ contains only the origin, and therefore 
\[
\expect_{\Phi}\sup_{\psibf \in T_\epsilon(0)} \abs{ Z_N(\psibf) } = \expect_{\Phi}\abs{ Z_N(0) } = 0,
\]
we obtain $\expect_{\Phi} \sup_{\psibf \in T_\epsilon(k) } \abs{ Z_N(\psibf) } < \sqrt{C_\epsilon} \sum_{i=1}^{k}\frac{\sqrt{  k A_1 + A_2 }}{2^{k/2}}$
as required.
\end{proof}

\begin{lemma}\label{lem:maxineq}(Maximal inequality)
Suppose $X_1, \dots, X_N$ are zero mean Gaussian random variables each with variance less than some positive constant $K$, then $\expect \sup_{n = 1, \dots, N} \abs{X_n} \leq 3  \sqrt{K \log N}$ where $\log N$ is the natural logarithm of $N$.
\end{lemma}
\begin{proof}
This result is well known, see for example~\cite[Section 3]{Pollard_asymp_empi_proc_1989}  
\end{proof}

\begin{lemma}\label{lem:metricentropy} (Covering numbers)
For $k \in \ints$ there exists a discrete set $T_\epsilon(k) \subset \reals^{m+1}$ with the property that, for every $\psibf \in B_\epsilon$, there is a $\psibf^* \in T_\epsilon(k)$ such that,
\[
d(\psibf,\psibf^*) = \sum_{n = 1}^N \big( f_{nN}(\psibf) - f_{nN}(\psibf^*) \big)^2 \leq \frac{C_\epsilon N}{2^k}.
\]
The number of elements in $T_\epsilon(k)$ is no more than $K_3 2^{(m+1)k}$ where $K_3$ is a positive constant, independent of $N$, $\epsilon$ and $k$.
\end{lemma}

Before we give the proof of this lemma we need some results from the literature on hyperplane arrangements and what are called \emph{$\epsilon$-cuttings}~\cite{Chazelle_discrepency_method_2000,Matousek_lect_disc_geom_2002}.  Let $H$ be a set of $m$-dimensional affine hyperplanes lying in $\reals^{m+1}$.  By \emph{affine} it is meant that  the hyperplanes need not pass through the origin.  For each hyperplane $h \in H$ let $D(h)$ and its complement $\bar{D}(h)$ be the corresponding half spaces of $\mathbb{R}^{m+1}$.  For a point $\xbf \in \mathbb{R}^{m+1}$, let
\begin{equation}\label{eq:bhsi}
b(h,\xbf) = \begin{cases} 
1 & \xbf \in D(h) \\
0 & \xbf \in \bar{D}(h).
\end{cases}
\end{equation}
Note that $b(h,\xbf)$ is piecewise constant in $\xbf$.  Two points $\xbf$ and $\ybf$ from $\reals^{m+1}$ are in the same halfspace of $h$ if and only if $b(h, \xbf) = b(h, \ybf)$.  So the pseudometric
\[
\sigma(\xbf,\ybf) = \sum_{h \in H} |b(h,\xbf) - b(h,\ybf)|
\]
is the number of hyperplanes in $H$ that pass between the points $\xbf$ and $\ybf$.  

The next theorem considers the partitioning of $\reals^{m+1}$ into subsets so that not too many hyperplanes intersect with any subset.  Proofs can be found in Theorem 5.1 on page 206 of~\cite{Chazelle_discrepency_method_2000} and also Theorem 6.5.3 on page 144 of~\cite{Matousek_lect_disc_geom_2002}.

\begin{theorem}\label{thm:epscutting}
There exists a constant $K$, independent of the set of hyperplanes $H$, such that for any positive real number $r$, we can partition $\reals^{m+1}$ into $K r^{m+1}$ generalised $(m+1)$-dimensional simplices with the property that no more than $\abs{H}/r$ hyperplanes from $H$ pass through the interior of any simplex.
\end{theorem}

By the phrase `There exists a constant $K$, independent of the set of hyperplanes $H$', it is meant that the constant $K$ is valid for every possible set of hyperplanes in $\reals^{m+1}$, regardless of the number of hyperplanes or their position and orientation.  A \emph{generalised} $(m+1)$-dimensional simplex is the region defined by the intersection of $m+2$ half spaces in $\reals^{m+1}$.  Note that a generalised simplex (unlike an ordinary simplex) can be unbounded.  For our purposes Theorem~\ref{thm:epscutting} is important because of the following corollary.

\begin{corollary}\label{cor:epscutting}
There exists a constant $K$, independent of the set of hyperplanes $H$, such that for every positive real number $r$ there is a discrete subset $T \subset \reals^{m+1}$ containing no more than $K r^{m+1}$ elements with the property that for every $\xbf \in \reals^{m+1}$ there exists $\ybf \in T$ with $\sigma(\xbf, \ybf) \leq \abs{H}/r$.
\end{corollary}
\begin{proof}
Let $C$ be the set of generalised simplices constructed according to Theorem~\ref{thm:epscutting}.  Define $T$ as a set containing precisely one point from the interior of each simplex in $C$.  Let $\xbf \in \reals^{m+1}$.  Since $b(h, \xbf)$ is piecewise constant for each $h \in H$, and since $C$ partitions $\reals^{m+1}$, there must exist a simplex $c \in C$ with a point $\zbf$ in its interior such that $b(h,\zbf) = b(h, \xbf)$ for all $h \in H$, and correspondingly $\sigma(\zbf, \xbf) = 0$.  Let $\ybf$ be the element from $T$ that is in the interior of $c$.  Since at most $\abs{H}/r$ hyperplanes cross the interior of $c$ there can be at most $\abs{H}/r$ hyperplanes between $\zbf$ and $\ybf$, and so $\sigma(\zbf, \ybf) \leq \abs{H}/r$.  Now $\sigma(\xbf, \ybf) < \sigma(\zbf, \xbf) + \sigma(\zbf, \ybf) \leq \frac{\abs{H}}{r}$ follows from the triangle inequality.
\end{proof}

The previous corollary ensures that we can \emph{cover} $\reals^{m+1}$ using $K r^{m+1}$ `balls' of radius $\abs{H}/r$ with respect to the pseudometric $\sigma(\xbf, \ybf)$.  The balls are placed at the positions defined by points in the set $T$, and these points could be anywhere in $\reals^{m+1}$.  The next corollary asserts that we can cover a subset of $\reals^{m+1}$, by placing the balls at points only within this subset.

\begin{corollary}\label{cor:epscuttingsubset}
Let $B$ be a subset of $\reals^{m+1}$.  There exists a constant $K$, independent of the set of hyperplanes $H$, such that for every positive real number $r$ there is a discrete subset $T_B \subset B$ containing no more than $K r^{m+1}$ elements with the property that for every $\xbf \in B$ there exists $\ybf \in T_B$ with $\sigma(\xbf, \ybf) \leq \abs{H}/r$.
\end{corollary}
\begin{proof}
Let $C$ be the set of generalised simplices constructed according to Theorem~\ref{thm:epscutting} and let $C_B$ be the subset of those indices that intersect $B$.  Let $T_B$ contain a point from $c \cap B$ for each simplex $c \in C_B$.  The proof now follows similarly to Corollary~\ref{cor:epscutting}.
\end{proof}

We are now ready to prove Lemma~\ref{lem:metricentropy}.

\begin{proof} (Lemma~\ref{lem:metricentropy})
Put $g_{nN}(\psibf) = (\tfrac{N}{n})^\ell f_{nN}(\psibf) = \round{ \Phi_n + p_{nN}(\psibf)}$, and let
\[
d_g(\psibf, \psibf^*) = \sum_{n=1}^{N} \big( g_{nN}(\psibf) - g_{nN}(\psibf^*) \big)^2.
\]
We have $d(\psibf, \psibf^*) \leq d_g(\psibf, \psibf^*)$, and so it suffices to prove the lemma with $d$ replaced by $d_g$.  
From~\eqref{eq:pnNsmall} it follows that $\abs{p_{nN}(\psibf)} \leq (m+1)\epsilon < 1$.  Since $\Phi_n \in [-\nicefrac{1}{2}, \nicefrac{1}{2})$, when $\Phi_n \geq 0$,
\[
g_{nN}(\psibf) = \begin{cases}
1 & p_{nN}(\psibf) \geq \nicefrac{1}{2} - \Phi_n \\
0 & \text{otherwise},
\end{cases}
\]
and when $\Phi_n < 0$,
\[
g_{nN}(\psibf) = \begin{cases}
-1 & p_{nN}(\psibf) < -\nicefrac{1}{2} - \Phi_n \\
0 & \text{otherwise}.
\end{cases}
\]
Thus, $( g_{nN}(\psibf) - g_{nN}(\psibf^*) )^2$ is either equal to one when $g_{nN}(\psibf) \neq g_{nN}(\psibf^*)$ or zero when $g_{nN}(\psibf) = g_{nN}(\psibf^*)$.  Now $g_{nN}(\psibf) \neq 0$ only if 
\[
\abs{\Phi_n} \geq \nicefrac{1}{2} - \abs{p_{nN}(\psibf)} \geq \nicefrac{1}{2} - (m+1)\epsilon,
\]
that is, only if $I_\epsilon(\Phi_n) = 1$.  Let $A = \{ n \in \{1, \dots, N\} \mid I_\epsilon(\Phi_n) = 1 \}$ be the subset of the indices where $I_\epsilon(\Phi_n) = 1$.  By definition the number of elements in $A$ is $C_\epsilon N$ (see~\eqref{eq:Cedefn}).  If both $\psibf$ and $\psibf^*$ are in $B_\epsilon$, then 
\[
(g_{nN}(\psibf) - g_{nN}(\psibf^*) )^2 \neq 0 
\]
only if $n \notin A$.  Thus,
\begin{align*}
d_g(\psibf, \psibf^*) = \sum_{n=1}^{N} \big( g_{nN}(\psibf) - g_{nN}(\psibf^*) \big)^2 = \sum_{n \in A} \big( g_{nN}(\psibf) - g_{nN}(\psibf^*) \big)^2.
\end{align*}

We now use Corollary~\ref{cor:epscuttingsubset}.  Let $h_n$ be the $m$ dimensional hyperplane in $\reals^{m+1}$ satisfying
\[
p_{nN}(\psibf) =  \sum_{i=0}^{m} \big(\tfrac{n}{N}\big)^i \psi_i = \tfrac{1}{2} \sign{\Phi_n} - \Phi_n
\]
where $\sign{\Phi_n}$ is equal to $1$ when $\Phi_n \geq 0$ and $-1$ otherwise.  The hyperplane $h_n$ divides $\reals^{m+1}$ into two halfspaces, $D(h_n)$ and its complement $\bar{D}(h_n)$.  If $\psibf$ and $\psibf^*$ are in the same halfspace, then $\abs{ b(h_n,\psibf) - b(h_n,\psibf^*) } = 0$ and $g_{nN}(\psibf) = g_{nN}(\psibf^*)$, and therefore $( g_{nN}(\psibf) - g_{nN}(\psibf^*) )^2 = 0$.  Otherwise, if $\psibf$ and $\psibf^*$ are in different halfspaces, then $\abs{ b(h_n,\psibf) - b(h_n,\psibf^*) } = 1$ and $g_{nN}(\psibf) \neq g_{nN}(\psibf^*)$, and therefore $( g_{nN}(\psibf) - g_{nN}(\psibf^*) )^2 = 1$.  Thus, 
\[
(g_{nN}(\psibf) - g_{nN}(\psibf^*) )^2 = \abs{ b(h_n,\psibf) - b(h_n,\psibf^*) }
\]
for $n = 1, \dots,  N$.  Let $H$ be the finite set of hyperplanes $\{ h_n, n \in A\}$ and observe that the number of hyperplanes is $\abs{H} = \abs{A} = C_\epsilon N$.  When both $\psibf$ and $\psibf^*$ are inside $B_\epsilon$, $d_g$ can be written as
\begin{align*}
d_g(\psibf, \psibf^*) &= \sum_{n\in A} \abs{ b(h_n,\psibf) - b(h_n,\psibf^*) } \\
&=  \sum_{h \in H} \abs{ b(h,\psibf) - b(h,\psibf^*) } = \sigma(\psibf, \psibf^*).
\end{align*}
That is, when both $\psibf, \psibf^* \in B_\epsilon$, $d_g(\psibf, \psibf^*)$ is the number of hyperplanes from $H$ that pass between the points $\psibf$ and $\psibf^*$.  

It follows from Corollary~\ref{cor:epscuttingsubset} that for any positive $r$ there exists a finite subset $T_B$ of $B_\epsilon$ containing at most $K_3 r^{m+1}$ elements, such that for every $\psibf \in B_\epsilon$ there is a $\psibf^* \in T_B$ with 
\[
d_g(\psibf, \psibf^*) = \sigma(\psibf, \psibf^*) \leq \frac{\abs{H}}{r} = \frac{\abs{A}}{r} = \frac{C_\epsilon N}{r}.
\]
Putting $r = 2^k$ and choosing $T_\epsilon(k) = T_B$ completes the proof.
\end{proof}

\bibliographystyle{IEEEbib}
\bibliography{bib} 

\begin{thebibliography}{10}

\bibitem{Hlawatsch_lin_quad_time_freq_spmag_1992}
F.~Hlawatsch and G.~F. Boudreaux-Bartels,
\newblock ``Linear and quadratic time-frequency signal representations,''
\newblock {\em IEEE Signal Processing Magazine}, vol. 9, no. 2, pp. 21--67,
  Apr. 1992.

\bibitem{Ridleyspeechpolyphase1989}
M.~D. Ridley,
\newblock {\em Speech time-frequency representations},
\newblock Kluwer Academic Publishers, Boston, 1989.

\bibitem{Suga_1975_bats_echolocation}
N.~Suga, J.~A. Simmons, and P.~H. Jen,
\newblock ``Peripheral specialization for fine analysis of doppler-shifted
  echoes in the auditory system of the {CF-FM} bat {Pteronotus Parnellii},''
\newblock {\em Journal of Experimental Biology}, vol. 63, pp. 161--192, 1975.

\bibitem{Moss_2005echolocation}
J.~A. Thomas, C.~F. Moss, M.~Vater, and P.~W. Moore,
\newblock ``Echolocation in bats and dolphins,''
\newblock {\em The Journal of the Acoustical Society of America}, vol. 118, pp.
  2755, 2005.

\bibitem{McKilliam2009IndentifiabliltyAliasingPolyphase}
R.~G. McKilliam and I.~V.~L. Clarkson,
\newblock ``Identifiability and aliasing in polynomial-phase signals,''
\newblock {\em IEEE Trans. Sig. Process.}, vol. 57, no. 11, pp. 4554--4557,
  Nov. 2009.

\bibitem{Hannan1973}
E.~J. Hannan,
\newblock ``The estimation of frequency,''
\newblock {\em Journal of Applied Probability}, vol. 10, no. 3, pp. 510--519,
  Sep 1973.

\bibitem{Quinn2001}
B.~G. Quinn and E.~J. Hannan,
\newblock {\em The Estimation and Tracking of Frequency},
\newblock Cambridge University Press, New York, 2001.

\bibitem{McKilliam_mean_dir_est_sq_arc_length2010}
R.~G. McKilliam, B.~G. Quinn, and I.~V.~L. Clarkson,
\newblock ``Direction estimation by minimum squared arc length,''
\newblock {\em IEEE Trans. Sig. Process.}, vol. 60, no. 5, pp. 2115--2124, May
  2012.

\bibitem{Peleg_DPT_1995}
S.~Peleg and B.~Friedlander,
\newblock ``The discrete polynomial-phase transform,''
\newblock {\em IEEE Trans. Sig. Process.}, vol. 43, no. 8, pp. 1901--1914, Aug.
  1995.

\bibitem{Farquharson_another_poly_est_2005}
M.~Farquharson, P.~O'Shea, and G.~Ledwich,
\newblock ``A computationally efficient technique for estimating the parameters
  of polynomial-phase signals from noisy observations,''
\newblock {\em IEEE Trans. Sig. Process.}, vol. 53, no. 8, pp. 3337--3342, Aug.
  2005.

\bibitem{Porat_asympt_HAF_DPT_1996}
B.~Porat and B.~Friedlander,
\newblock ``Asymptotic statistical analysis of the high-order ambiguity
  function for parameter estimation of polynomial-phase signals,''
\newblock {\em IEEE Trans. Inform. Theory}, vol. 42, no. 3, pp. 995--1001, May
  1996.

\bibitem{Kitchen_polyphase_unwrapping_1994}
J.~Kitchen,
\newblock ``A method for estimating the coefficients of a polynomial phase
  signal,''
\newblock {\em Signal Processing}, vol. 37, no. 1, pp. 463--470, Jun. 1994.

\bibitem{Morelande_bayes_unwrapping_2009_tsp}
M.~R. Morelande,
\newblock ``Parameter estimation of phase-modulated signals using {Bayesian}
  unwrapping,''
\newblock {\em IEEE Trans. Sig. Process.}, vol. 57, no. 11, pp. 4209--4219,
  Nov. 2009.

\bibitem{McKilliam2010thesis}
R.~G. McKilliam,
\newblock {\em Lattice theory, circular statistics and polynomial phase
  signals},
\newblock Ph.D. thesis, University of Queensland, Australia, December 2010.

\bibitem{McKilliam2009asilomar_polyest_lattice}
R.~G. McKilliam, I.~V.~L. Clarkson, B.~G. Quinn, and B.~Moran,
\newblock ``Polynomial-phase estimation, phase unwrapping and the nearest
  lattice point problem,''
\newblock in {\em Forty-Third Asilomar Conference on Signals, Systems and
  Computers}, Asilomar, California, Nov. 2009, pp. 493--495, IEEE.

\bibitem{Agrell2002}
E.~Agrell, T.~Eriksson, A.~Vardy, and K.~Zeger,
\newblock ``Closest point search in lattices,''
\newblock {\em IEEE Trans. Inform. Theory}, vol. 48, no. 8, pp. 2201--2214,
  Aug. 2002.

\bibitem{McKilliam_polyphase_est_icassp_2011}
R.~G. McKilliam, B.~G. Quinn, I.~V.~L. Clarkson, and B.~Moran,
\newblock ``The asymptotic properties of polynomial phase estimation by least
  squares phase unwrapping,''
\newblock in {\em Proc. Internat. Conf. Acoust. Spe. Sig. Process.}, Prauge,
  Czech Republic, May 2011, pp. 3592--3595, IEEE.

\bibitem{Erdos_on_some_sequence_of_integers1936}
P.~Erd\"{o}s and P.~Turan,
\newblock ``On some sequences of integers,''
\newblock {\em J. London Math. Soc.}, vol. 11, pp. 261--264, 1936.

\bibitem{Szemeredi_setint_no_k_arth1975}
E.~Szemer\'{e}di,
\newblock ``On sets of integers containing no $k$ elements in arithmetic
  progression,''
\newblock {\em Acta Arith.}, vol. 27, pp. 299--345, 1975.

\bibitem{Gowers_new_proof2001}
W.~T. Gowers,
\newblock ``A new proof of {Szemer\'{e}di's} theorem,''
\newblock {\em Geom. Func. Anal.}, vol. 11, pp. 465–588, 2001.

\bibitem{Pollard_new_ways_clts_1986}
D.~Pollard,
\newblock ``New ways to prove central limit theorems,''
\newblock {\em Econometric Theory}, vol. 1, no. 3, pp. 295--313, Dec. 1985.

\bibitem{Pollard_asymp_empi_proc_1989}
D.~Pollard,
\newblock ``Asymptotics by empirical processes,''
\newblock {\em Statistical Science}, vol. 4, no. 4, pp. 341--354, 1989.

\bibitem{van2009empirical}
S.~A. van~de Geer,
\newblock {\em Empirical Processes in M-Estimation},
\newblock Cambridge Series in Statistical and Probabilistic Mathematics.
  Cambridge University Press, 2009.

\bibitem{Dudley_unif_central_lim_th_1999}
R.~M. Dudley,
\newblock {\em Uniform central limit theorems},
\newblock Cambridge University Press, 1999.

\bibitem{Chazelle_discrepency_method_2000}
B.~Chazelle,
\newblock {\em The Discrepency Method: Randomness and Complexity},
\newblock Cambridge University Press, 2000.

\bibitem{Matousek_lect_disc_geom_2002}
J.~Matousek,
\newblock {\em Lectures on Discrete Geometry}, vol. 212 of {\em Graduate Texts
  in Mathematics},
\newblock Springer, 2002.

\bibitem{Nico_phaseunwrappingSAR_2000}
G.~Nico, G.. Palubinskas, and M.~Datcu,
\newblock ``Bayesian approaches to phase unwrapping: theoretical study,''
\newblock {\em IEEE Trans. Sig. Process.}, vol. 48, no. 9, pp. 2545--2556, Sep.
  2000.

\bibitem{Friedlander_PD_phaseunwrapping_1996}
B.~Friedlander and J.M. Francos,
\newblock ``Model based phase unwrapping of {2-D} signals,''
\newblock {\em IEEE Trans. Sig. Process.}, vol. 44, no. 12, pp. 2999--3007,
  Dec. 1996.

\bibitem{Cassels_geom_numbers_1997}
J.~W.~S. Cassels,
\newblock {\em An Introduction to the Geometry of Numbers},
\newblock Springer, 1997.

\bibitem{cahen_integer-valued_1997}
P.~J. Cahen and J.~L. Chabert,
\newblock {\em Integer-valued Polynomials},
\newblock American Mathematical Society, Providence RI, July 1997.

\bibitem{Mardia_directional_statistics}
K.~V. Mardia and P.~Jupp,
\newblock {\em Directional Statistics},
\newblock John Wiley \& Sons, 2nd edition, 2000.

\bibitem{Fisher1993}
N.~I. Fisher,
\newblock {\em {Statistical analysis of circular data}},
\newblock Cambridge University Press, 1993.

\bibitem{Bhattacharya_int_ext_means_2003}
R.~N. Bhattacharya and V.~Patrangenaru,
\newblock ``Large sample theory of intrinsic and extrinsic sample means on
  manifolds {I},''
\newblock {\em Annals of Statistics}, vol. 31, no. 1, pp. 1--29, 2003.

\bibitem{Bhattacharya_int_ext_means_2005}
R.~N. Bhattacharya and V.~Patrangenaru,
\newblock ``Large sample theory of intrinsic and extrinsic sample means on
  manifolds {II},''
\newblock {\em Annals of Statistics}, vol. 33, no. 3, pp. 1225--1259, 2005.

\bibitem{McKilliamFrequencyEstimationByPhaseUnwrapping2009}
R.~G. McKilliam, B.~G. Quinn, I.~V.~L. Clarkson, and B.~Moran,
\newblock ``Frequency estimation by phase unwrapping,''
\newblock {\em IEEE Trans. Sig. Process.}, vol. 58, no. 6, pp. 2953--2963, June
  2010.

\bibitem{Hotz_circle_means_2011}
T.~Hotz and S.~Huckemann,
\newblock ``Intrinsic means on the circle: Uniqueness, locus and asymptotics,''
\newblock http://arxiv.org/abs/1108.2141, 2011.

\bibitem{Pollard_conv_stat_proc_1984}
D.~Pollard,
\newblock {\em Convergence of Stochastic Processes},
\newblock Springer-Verlag, New York, 1984.

\bibitem{Jordan_Calculus_of_finite_difference_1965}
C.~Jordan,
\newblock {\em Calculus of finite differences},
\newblock Chelsea Publishing Company, New York, N.Y., 1965.

\bibitem{vonMises_diff_stats_1947}
R.~von Mises,
\newblock ``On the asymptotic distribution of differentiable statistical
  functions,''
\newblock {\em Annals of Mathematical Statistics}, vol. 18, no. 3, pp.
  309--348, Sep. 1947.

\bibitem{Pohst_sphere_decoder_1981}
M.~Pohst,
\newblock ``On the computation of lattice vectors of minimal length, successive
  minima and reduced bases with applications,''
\newblock {\em SIGSAM Bull.}, vol. 15, no. 1, pp. 37--44, 1981.

\bibitem{Viterbo_sphere_decoder_1999}
E.~Viterbo and J.~Boutros,
\newblock ``A universal lattice code decoder for fading channels,''
\newblock {\em IEEE Trans. Inform. Theory}, vol. 45, no. 5, pp. 1639--1642,
  Jul. 1999.

\bibitem{Zhan2006_K_best_sphere_decoder}
Z.~Guo and P.~Nilsson,
\newblock ``Algorithm and implementation of the {K}-best sphere decoding for
  {MIMO} detection,''
\newblock {\em IEEE J. Sel. Areas Commun.}, vol. 24, no. 3, pp. 491 -- 503,
  March 2006.

\bibitem{Peleg1991_CRB_PPS_1991}
S.~Peleg and B.~Porat,
\newblock ``The {Cram\`{e}r-Rao} lower bound for signals with constant
  amplitude and polynomial phase,''
\newblock {\em IEEE Trans. Sig. Process.}, vol. 39, no. 3, pp. 749--752, Mar.
  1991.

\bibitem{Hoeffding_inequality_1963}
W.~Hoeffding,
\newblock ``Probability inequalities for sums of bounded random variables,''
\newblock {\em J. Amer. Statist. Assoc.}, vol. 58, no. 301, pp. 13--30, 1963.

\bibitem{Billingsley1999_convergence_of_probability_measures}
P.~Billingsley,
\newblock {\em Convergence of probability measures},
\newblock John Wiley \& Sons, 2nd edition, 1999.

\bibitem{Shorak_emp_proc_stat_2009}
G.~R. Shorack and J.~A. Wellner,
\newblock {\em Empirical Processes with Applications to Statistics},
\newblock Classics in Applied Mathematics. Cambridge University Press, Sep.
  2009.

\bibitem{Gine_Zinn_symmetrisation_1984}
E.~Gin\'{e} and J.~Zinn,
\newblock ``Some limit theorems for empirical processes,''
\newblock {\em Annals of Probability}, vol. 12, no. 4, pp. 929--989, 1984.

\bibitem{Ossiander_clt_bracketing_1984}
M.~Ossiander,
\newblock ``A central limit theorem under metric entropy with $l_2$
  bracketing,''
\newblock {\em Annals of Probability}, vol. 15, no. 3, pp. 897--891, 1984.

\end{thebibliography}

\end{document}